\documentclass[preprint]{elsarticle}
\usepackage[top=1in,bottom=1in,left=1in,right=1in]{geometry}

\usepackage[utf8x]{inputenc}
\usepackage[section]{algorithm}
\usepackage{algpseudocode}
\usepackage{enumitem}
\usepackage{amsfonts}
\usepackage{amsmath}
\usepackage{amsthm}
\usepackage{verbatim}
\usepackage{graphicx}
\usepackage{float}
\usepackage{caption}
\usepackage{subcaption}
\usepackage[bookmarks]{hyperref}

\newtheorem{definition}{Definition}

\newtheorem{theorem}{Theorem}
\newtheorem{lemma}[theorem]{Lemma}
\newtheorem{corollary}[theorem]{Corollary}

\graphicspath{ {./images/} }

\date{}

\begin{document}

\title{Approximability of Guarding Weak Visibility Polygons
\footnote{An extended abstract \cite{VGinWVP} appeared in the Proceedings of the International Conference on Algorithms and Discrete Applied Mathematics (CALDAM),
Lecture Notes in Computer Science (LNCS) 8959, pp. 45-57, Springer, 2015.}}

\author[add1]{Pritam Bhattacharya} 
  \ead{pritam.bhattacharya@cse.iitkgp.ernet.in}
\author[add2]{Subir Kumar Ghosh}
  \ead{ghosh@tifr.res.in}
\author[add3]{Bodhayan Roy}
  \ead{bodhayan.roy@gmail.com}

  \address[add1]{Department of Computer Science and Engineering, Indian Institute of Technology, Kharagpur - 721302, India.}
  \address[add2]{Department of Computer Science, Ramakrishna Mission Vivekananda University, Belur Math, Howrah - 711202, India.}
  \address[add3]{Department of Computer Science and Engineering, Indian Institute of Technology, Powai, Mumbai - 400076, India.}

\begin{abstract}
The art gallery problem enquires about the least number of guards that are sufficient to ensure that an art gallery, represented by a 
polygon $P$, is fully guarded. In 1998, the problems of finding the minimum number of point guards, vertex guards, and edge guards 
required to guard $P$ were shown to be APX-hard by Eidenbenz, Widmayer and Stamm. In 1987, Ghosh presented approximation algorithms 
for vertex guards and edge guards that achieved a ratio of $\mathcal{O}(\log n)$, which was improved upto $\mathcal{O}(\log\log OPT)$ 
by King and Kirkpatrick in 2011. It has been conjectured that constant-factor approximation algorithms exist for these problems. We 
settle the conjecture for the special class of polygons that are weakly visible from an edge and contain no holes by presenting a 
6-approximation algorithm for finding the minimum number of vertex guards that runs in $\mathcal{O}(n^2)$ time. On the other hand,
for weak visibility polygons with holes, we present a reduction from the Set Cover problem to show that there cannot exist a polynomial 
time algorithm for the vertex guard problem with an approximation ratio better than $((1−\epsilon)/12)\ln n$ for any $\epsilon>0$, 
unless NP~=~P.  We also show that, for the special class of polygons without holes that are orthogonal as well as weakly visible from 
an edge, the approximation ratio can be improved to 3. Finally, we consider the point guard problem and show that it is NP-hard in the
case of polygons weakly visible from an edge.
\end{abstract}


\maketitle

\section{Introduction}
\label{intro}

\subsection{The art gallery problem and its variants}
\label{agp}
The art gallery problem enquires about the least number of guards that are sufficient to ensure that an art gallery 
(represented by a polygon $P$) is fully guarded, assuming that a guard’s field of view covers 360\textdegree\:as well 
as an unbounded distance. This problem was first posed by Victor Klee in a conference in 1973, and in the course of 
time, it has turned into one of the most investigated problems in computational geometry. \\

A \emph{polygon} $P$ is defined to be a closed region in the plane bounded by a finite set of line segments, called edges of 
$P$, such that, between any two points of $P$, there exists a path which does not intersect any edge of $P$. If the boundary 
of a polygon $P$ consists of two or more cycles, then $P$ is called a \emph{polygon with holes} (see Figure \ref{sp_a}).
Otherwise, $P$ is called a \emph{simple polygon} or a \emph{polygon without holes} (see Figure \ref{sp_b}). \\

An art gallery can be viewed as an $n$-sided polygon $P$ (with or without holes) and guards as points inside $P$.
Any point $z \in P$ is said to be \emph{visible} from a guard $g$ if the line segment ${z g}$ does not intersect 
the exterior of $P$ (see Figure \ref{sp_a} and Figure \ref{sp_b}). In general, guards may be placed anywhere inside $P$. 
If the guards are allowed to be placed only on vertices of $P$, they are called \emph{vertex guards}. If there is 
no such restriction, guards are called \emph{point guards}. Point and vertex guards together are also referred to 
as \emph{stationary guards}. If guards are allowed to patrol along a line segment inside $P$, they are called 
\emph{mobile guards}. If they are allowed to patrol only along the edges of $P$, they are called \emph{edge guards}. 
\cite{Ghosh_2007,O'Rourke_1987} \\ 


\begin{figure}[H]
\begin{minipage}{.39\textwidth}
\centerline{\includegraphics[width=\textwidth]{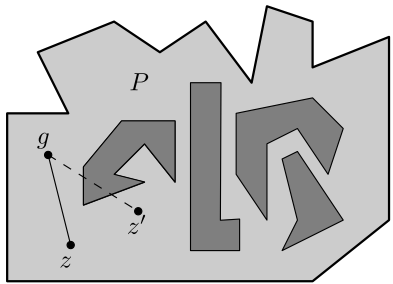}}
\caption{Polygon with holes}
\label{sp_a}
\end{minipage}
\hspace*{.17\textwidth}
\begin{minipage}{.39\textwidth}
\centerline{\includegraphics[width=\textwidth]{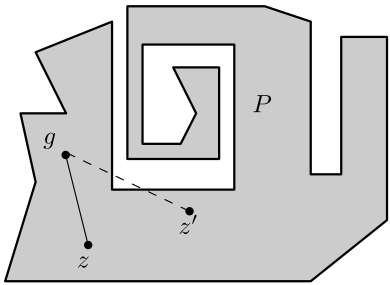}}
\caption{Polygon without holes}
\label{sp_b}
\end{minipage}
\end{figure}

In 1975, Chvátal \cite{Chvatal_1975} showed that $\lfloor\frac{n}{3}\rfloor$ stationary guards are sufficient and sometimes necessary 
(see Figure \ref{sg_a}) for guarding a simple polygon. In 1978, Fisk \cite{Fisk_1978} presented a simpler and more elegant proof of 
this result. For a simple orthogonal polygon, whose edges are either horizontal or vertical, Kahn et al. \cite{KKK_1983} and also 
O’Rourke \cite{O'Rourke_1983} showed that $\lfloor\frac{n}{4}\rfloor$ stationary guards are sufficient and sometimes necessary (see Figure \ref{sg_b}). 

\begin{figure}[H]
\begin{minipage}{.41\textwidth}
\centerline{\includegraphics[width=\textwidth]{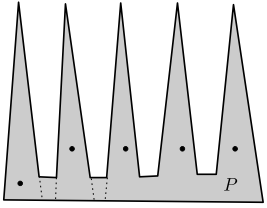}}
\caption{A polygon where $\lfloor\frac{n}{3}\rfloor$ stationary guards are necessary.}
\label{sg_a}
\end{minipage}
\hspace*{.11\textwidth}
\begin{minipage}{.45\textwidth}
\centerline{\includegraphics[width=\textwidth]{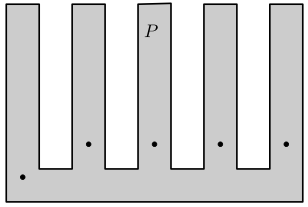}}
\caption{A polygon where $\lfloor\frac{n}{4}\rfloor$ stationary guards  are necessary.}
\label{sg_b}
\end{minipage}
\end{figure}


\subsection{Related hardness and approximation results}
\label{rhar}
The decision version of the art gallery problem is to determine, given a polygon $P$ and a number $k$ as input, whether 
the polygon $P$ can be guarded with $k$ or fewer guards. The problem was first proved to be NP-complete for polygons with 
holes by O’Rourke and Supowit \cite{RS_1983}. For guarding simple polygons, it was proved to be NP-complete for vertex guards 
by Lee and Lin \cite{LL_1986}, and their proof was generalized to work for point guards by Aggarwal \cite{Aggarwal_1984}. 
The problem is NP-hard even for simple orthogonal polygons as shown by Katz and Roisman \cite{KR_2008} and Schuchardt and 
Hecker \cite{SH_1995}. Each one of these hardness results hold irrespective of whether we are dealing with vertex guards, edge 
guards, or point guards. \\

\vspace*{-.4em}
In 1987, Ghosh \cite{Ghosh_1987,Ghosh_2010} provided an $\mathcal{O}(\log n)$-approximation algorithm for the case of vertex and 
edge guards by discretizing the input polygon and treating it as an instance of the Set Cover problem. In fact, applying methods 
for the Set Cover problem developed after Ghosh’s algorithm, it is easy to obtain an approximation factor of $\mathcal{O}(\log OPT)$ 
for vertex guarding simple polygons or $\mathcal{O}(\log h \log OPT)$ for vertex guarding a polygon with $h$ holes. Deshpande et al. 
\cite{DKDS_2007} obtained an approximation factor of $\mathcal{O}(\log OPT)$ for point guards or perimeter guards by developing a 
sophisticated discretization method that runs in pseudopolynomial time. Efrat and Har-Peled \cite{EH_2006} provided a randomized 
algorithm with the same approximation ratio that runs in fully polynomial expected time. For guarding simple polygons using vertex
guards and perimeter guards, King and Kirkpatrick \cite{KK_2011} obtained an approximation ratio of $\mathcal{O}(\log\log OPT)$ in 2011. \\

\vspace*{-.4em}
In 1998, Eidenbenz, Stamm and Widmayer \cite{ESW_1998,ESW_2001} proved that the problem is APX-complete, implying that an approximation 
ratio better than a fixed constant cannot be achieved unless P=NP. They also proved that if the input polygon is allowed to 
contain holes, then there cannot exist a polynomial time algorithm for the problem with an approximation ratio better than 
$((1−\epsilon)/12)\ln n$ for any $\epsilon>0$, unless NP $\subseteq$ TIME($n^{\mathcal{O}(\log\log n)}$). Contrastingly, in the 
case of simple polygons without holes, the existence of a constant-factor approximation algorithm for vertex guards and edge 
guards has been conjectured by Ghosh \cite{Ghosh_1987,GhoshW_2010} since 1987. However, this conjecture has not yet been settled 
even for special classes of polygons such as rectilinear, weak visibility, or L-R polygons. 

\subsection{Our contributions}
\label{contributions}
A polygon $P$ is said to be a \emph{weak visibility polygon} if every point in $P$ is visible from some point of an edge \cite{Ghosh_2007}. 
In Section \ref{algo}, we present a 6-approximation algorithm, which has running time $\mathcal{O}(n^2)$, for vertex guarding polygons that 
are weakly visible from an edge and contain no holes. This result can be viewed as a step forward towards solving Ghosh's conjecture for a 
special class of polygons. Then, in Section \ref{reduction}, by presenting a reduction from Set Cover we show that, for the special class of 
polygons containing holes that are weakly visible from an edge, there cannot exist a polynomial time algorithm for the vertex guard problem 
with an approximation ratio better than $((1−\epsilon)/12)\ln n$ for any $\epsilon>0$, unless NP~=~P. 
Next, in Section \ref{algo3}, we show that, for the special class of polygons without holes that are orthogonal as well as weakly visible from 
an edge, the approximation ratio can be improved to 3. Finally, in Section \ref{pgwvp}, we consider the point guard problem in weak visibility 
polygons and prove that it is NP-hard by showing a reduction from the decision version of the minimum line cover problem.

\section{Placement of vertex guards in weak visibility polygons}
\label{algo}

Let $P$ be a simple polygon. If there exists an edge ${uv}$ in $P$ (where $u$ is the next clockwise vertex of $v$) such that $P$ is 
weakly visible from ${uv}$, then it can be located in $\mathcal{O}(n^2)$ time \cite{AT_1981,Ghosh_1993}. Henceforth, we assume that 
such an edge ${uv}$ has been located. Let $bd_c(p,q)$ (or, $bd_{cc}(p,q)$) denote the clockwise (respectively, counterclockwise) 
boundary of $P$ from a vertex $p$ to another vertex $q$. Note that, by definition, $bd_c(p,q) = bd_{cc}(q,p)$. The \emph{visibility polygon} 
of $P$ from a point $z$, denoted by $VP(z)$, is defined to be the set of all points in $P$ that are visible from $z$. In other words, 
$ VP(z) = \{ q \in P : q \mbox{\: is visible from \:} z \} $. 

\begin{figure}[H]
  \centerline{\includegraphics[width=0.59\textwidth]{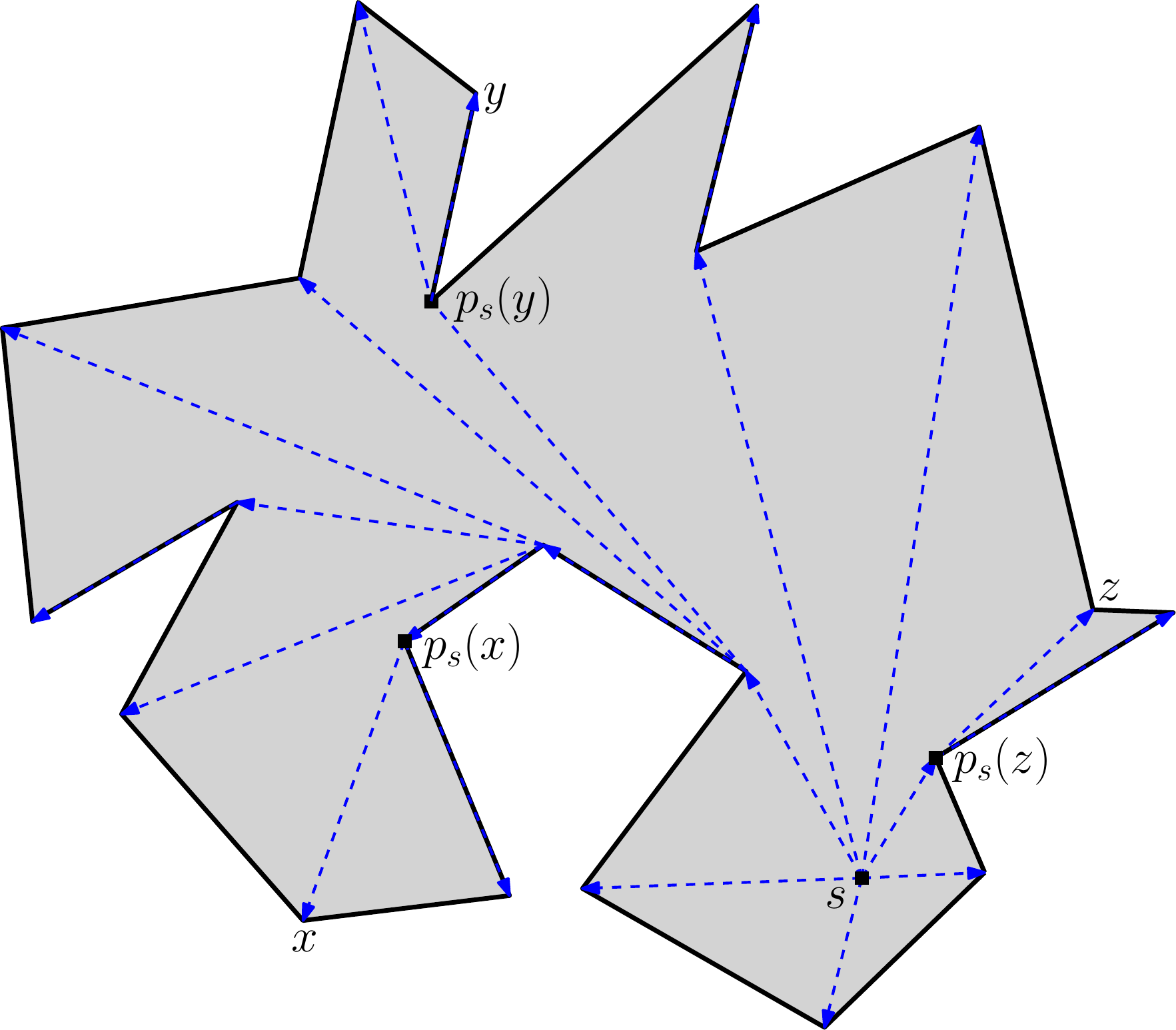}}
  \caption{Euclidean shortest path tree rooted at $s$. 
           The parents of vertices $x$, $y$ and $z$ in $SPT(s)$ are marked as $p_s(x)$, $p_s(y)$ and $p_s(z)$ respectively.}
  \label{spt}
\end{figure}

The \emph{shortest path tree} of $P$ rooted at a vertex $s$ of $P$, denoted by $SPT(s)$, is the union of Euclidean shortest paths from 
$s$ to all the vertices of $P$ (see Figure \ref{spt}. This union of paths is a planar tree, rooted at $r$, which has $n$ nodes, namely 
the vertices of $P$. 
For every vertex $x$ of $P$, let $p_u(x)$ and $p_v(x)$ denote the parent of $x$ in $SPT(u)$ and $SPT(v)$ respectively. In the same way, 
for every interior point $y$ of $P$, let $p_u(y)$ and $p_v(y)$ denote the vertex of $P$ next to $y$ in the Euclidean shortest path to $y$ 
from $u$ and $v$ respectively. 

\subsection{Guarding all vertices of a polygon}
\label{algo1}
Suppose a guard is placed on each non-leaf vertex of $SPT(u)$ and $SPT(v)$. It is obvious that these guards see all points of $P$.
However, the number of guards required may be very large compared to the size of an optimal guarding set. In order to reduce the number 
of guards, placing guards on every non-leaf vertex should be avoided. Let $A$ be a subset of vertices of $P$. Let $S_A$ denote the set 
which consists of the parents $p_u(z)$ and $p_v(z)$ of every vertex $z \in A$. Then, $A$ should be chosen such that all vertices of $P$ 
are visible from guards placed at vertices of $S_A$. We present a method for choosing $A$ and $S_A$ as follows:-
\begin{algorithm}[H]
\caption{An $\mathcal{O}(n^2)$-algorithm for computing a guard set $S_A$ for all vertices of $P$}
\label{VG_naive}
\begin{algorithmic}[1]
\State Compute $SPT(u)$ and $SPT(v)$ \label{VG_naive:1} 
\State Initialize all the vertices of $P$ as unmarked \label{VG_naive:2}
\State Initialize $A \leftarrow \emptyset$, $S_A \leftarrow \emptyset$ and $z \leftarrow u$ \label{VG_naive:3}
\While{$z \neq v$} \label{VG_naive:4}
\State $z \leftarrow$ the vertex next to $z$ in clockwise order on $bd_c(u,v)$ \label{VG_naive:5}
\If{$z$ is unmarked} \label{VG_naive:6}
\State $A \leftarrow A \cup \{z\}$ and $S_A \leftarrow S_A \cup \{p_u(z),p_v(z)\}$ \label{VG_naive:7}
\State Place guards on $p_u(z)$ and $p_v(z)$ \label{VG_naive:8}
\State Mark all vertices of $P$ that become visible from $p_u(z)$ or $p_v(z)$ \label{VG_naive:9}
\EndIf \label{VG_naive:10}
\EndWhile \label{VG_naive:11}
\State \Return the guard set $S_A$ \label{VG_naive:12}
\end{algorithmic}
\end{algorithm}

Now, assume a special condition such that for every vertex $z \in A$, all vertices of $bd_c(p_u(z),p_v(z))$ are visible from $p_u(z)$ or $p_v(z)$. 
We prove that, in such a situation, $|S_A| \leq 2|S_{opt}|$, where $S_{opt}$ denotes an optimal vertex guard set.  

\begin{lemma} \label{l1}
Any guard $g \in S_{opt}$ that sees vertex $z$ of $P$ must lie on $bd_c(p_u(z),p_v(z))$. 
\end{lemma}

\begin{proof}
Since $p_u(z)$ is the parent of $z$ in $SPT(u)$, $z$ cannot be visible from any vertex of $bd_c(u,p_u(z))$, except $p_u(z)$. Similarly, 
since $p_v(z)$ is the parent of $z$ in $SPT(v)$, $z$ cannot be visible from any vertex of $bd_{cc}(v,p_v(z))$, except $p_v(z)$. Hence, 
any guard $g \in S_{opt}$ that sees $z$ must lie on $bd_c(p_u(z),p_v(z))$. 
\end{proof}

\begin{lemma} \label{l2}
Let $z$ be a vertex of $P$ such that all vertices of $bd_c(p_u(z),p_v(z))$ are visible from $p_u(z)$ or $p_v(z)$.
For every vertex $x$ lying on $bd_c(p_u(z),p_v(z))$, if $x$ sees a vertex $q$ of $P$, then $q$ must also be visible from $p_u(z)$ or $p_v(z)$. 
\end{lemma}

\begin{proof}
If $q$ lies on $bd_c(p_u(z),p_v(z))$, then it is visible from from $p_u(z)$ or $p_v(z)$ by assumption. So, consider the case where $q$ lies on 
$bd_{cc}(p_u(z),p_v(z))$. Now, either $q$ lies on $bd_c(u,p_u(z))$ or $q$ lies on $bd_{cc}(v,p_v(z))$. In the former case, 
if $bd_{cc}(q,p_v(z))$ intersects the segment $q p_v(z)$, then $q$ or $p_v(z)$ is not weakly visible from $uv$ (see Figure \ref{lemma2}). Moreover, 
no other portion of the boundary can intersect $q p_v(z)$ since $qx$ and $z p_v(z)$ are internal segments. Hence, $q$ must be visible from $p_v(z)$. 
Analogously, if $q$ lies on $bd_{cc}(v,p_v(z))$, $q$ must be visible from $p_u(z)$.  
\end{proof}


\begin{lemma} \label{l3}
Assume that every vertex $z \in A$ is such that every vertex of $bd_c(p_u(z),p_v(z))$ is visible from $p_u(z)$ or $p_v(z)$. Then, $|A| \leq |S_{opt}|$. 
\end{lemma} 

\begin{proof}
Assume on the contrary that $|A| > |S_{opt}|$. This implies that Algorithm \ref{VG_naive} includes two distinct vertices $z_1$ and $z_2$ 
belonging to $A$ which are both visible from a single guard $g \in S_{opt}$. Moreover, it follows from Lemma \ref{l1} that $g$ must lie 
on $bd_c(p_u(z_1),p_v(z_1))$. Without loss of generality, let us assume that vertex $z_1$ is added to $A$ before $z_2$ by Algorithm 
\ref{VG_naive}. In that case, Algorithm \ref{VG_naive} places guards at $p_u(z_1)$ and $p_v(z_1)$. Now, as vertex $z_2$ is visible from 
$g$, it follows from Lemma \ref{l2} that $z_2$ is also visible from $p_u(z_1)$ or $p_v(z_1)$. Therefore, $z_2$ is already marked, and 
hence, Algorithm \ref{VG_naive} does not include $z_2$ in $A$, which is a contradiction.
\end{proof}

\begin{lemma} \label{l4}
$|S_A| = 2|A|$.
\end{lemma}

\begin{proof}
For every $z \in A$, since Algorithm \ref{VG_naive} includes both the parents $p_u(z)$ and $p_v(z)$ of $z$ in $S_A$, it is clear that 
$|S_A| \leq 2|A|$. If both the parents of every $z \in A$ are distinct, then $|S_A| = 2|A|$. Otherwise, there exists two distinct vertices 
$z_1$ and $z_2$ in $A$ that share a common parent, say $p$. Without loss of generality, let us assume that vertex $z_1$ is added to $A$ 
before $z_2$ by Algorithm \ref{VG_naive}. In that case, Algorithm \ref{VG_naive} places a guard at $p$, which results in $z_2$ getting 
marked. Thus, Algorithm \ref{VG_naive} cannot include $z_2$ in $A$, which is a contradiction. Hence, it must be the case that $|S_A| = 2|A|$.
\end{proof}

\begin{lemma} \label{l5}
If every vertex $z \in A$ is such that all vertices of $bd_c(p_u(z),p_v(z))$ are visible from $p_u(z)$ or $p_v(z)$, then $|S_A| \leq 2|S_{opt}|$. 
\end{lemma}

\begin{proof}
By Lemma \ref{l4}, $|S_A| = 2|A|$. By Lemma \ref{l3}, $|A| \leq |S_{opt}|$. So, $|S_A| = 2|A| \leq 2|S_{opt}|$.
\end{proof}


\begin{figure}[H]
\begin{minipage}{.36\textwidth}
  \centerline{\includegraphics[width=\textwidth]{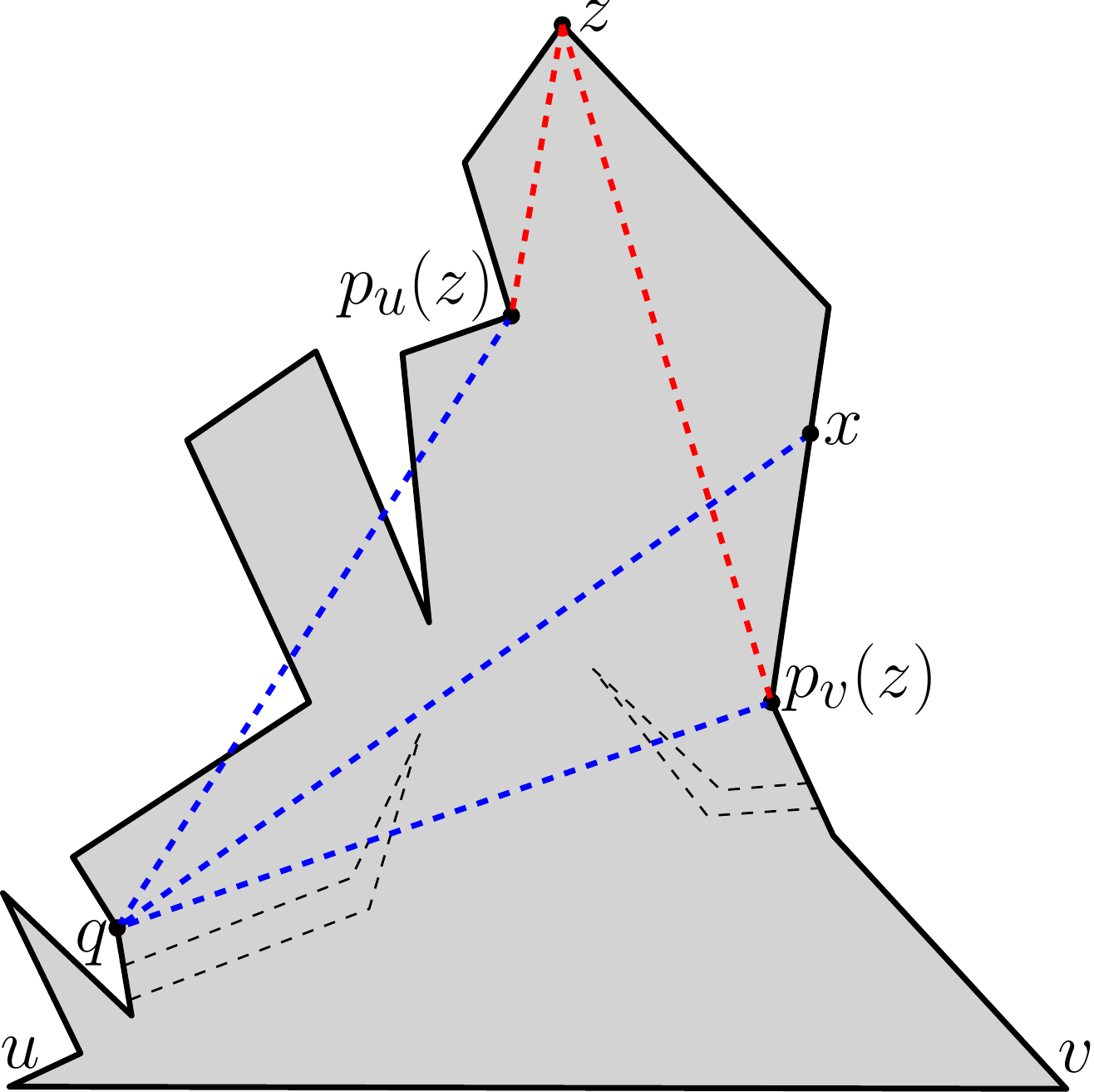}}
  \caption{Case in Lemma \ref{l2} where the segment ${q p_v(z)}$ is intersected by $bd_c(u,p_u(z))$.}
  \label{lemma2}
\end{minipage}
\quad
\begin{minipage}{.61\textwidth}
  \centerline{\includegraphics[width=\textwidth]{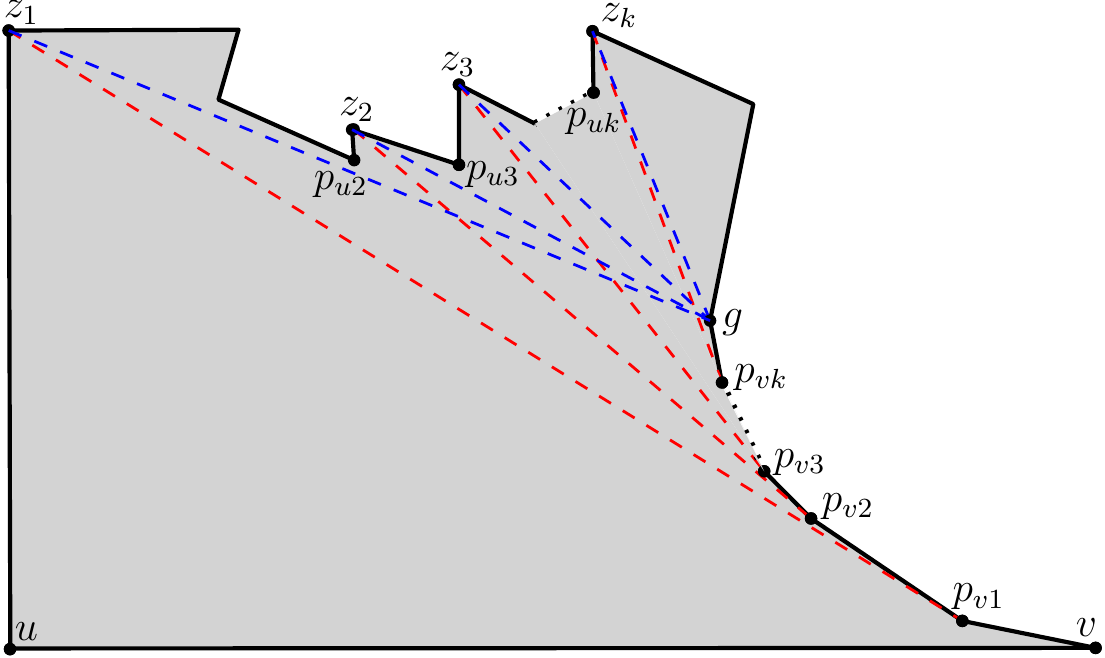}}
  \caption{An instance where the guard set $S_A$ computed by Algorithm \ref{VG_naive} is arbitrarily large compared to an optimal guard set $S_{opt}$.}
  \label{ce}
\end{minipage}
\end{figure}
 
The above bound does not hold if there exists $z \in A$ such that some vertices of $bd_c(p_u(z),p_v(z))$ are not visible from $p_u(z)$ or $p_v(z)$.
Consider Figure \ref{ce}. For each $i \in \{1,2,\dots,k-1\}$, $z_{i+1}$ is not visible from $p_u(z_i)$ or $p_v(z_i)$, which forces Algorithm \ref{VG_naive} 
to place guards at $p_u(z_{i+1})$ and $p_v(z_{i+1})$. Therefore, Algorithm \ref{VG_naive} includes $z_1,z_2,z_3,\dots,z_k$ in $A$ and places a total 
of $2k$ guards at vertices $u,p_{v1},p_{u2},p_{v2},\dots,p_{uk},p_{vk}$. However, all vertices of $P$ are visible from just two guards placed at $u$ 
and $g$. Hence, $|S_A|=2k$ whereas $|S_{opt}|=2$. Since the construction in Figure \ref{ce} can be extended for any arbitrary integer $k$, $|S_A|$ can 
be arbitrarily large compared to $|S_{opt}|$. So, we present a new algorithm which gives us a 4-approximation. \\ 

In the new algorithm, $bd_c(u,v)$ is scanned to identify a set of unmarked vertices, denoted as $B$, such that all vertices of $P$ are visible 
from guards in $S_B = \{p_u(z) | z \in B\} \cup \{p_v(z) | z \in B\}$. However, unlike the previous algorithm (see Algorithm \ref{VG_naive}), 
the new algorithm  (see Algorithm \ref{VG_4approx}) does not blindly include in $B$ every next unmarked vertex that it encounters during the scan.
During the scan, if $z$ denotes the current unmarked vertex being considered, then it may either choose to include $z$ in $B$ or skip ahead to
the next unmarked vertex along the scan depending on certain properties of $z$. At the end of each iteration of the outer while-loop (running from 
line \ref{VG_4approx:4} to line \ref{VG_4approx:22}), Algorithm \ref{VG_4approx} maintains the invariant that, for every unmarked vertex $y$ of 
$bd_c(u,z)$ (excluding $z$), $p_u(y)$ and $p_v(y)$ see all unmarked vertices of $bd_c(p_u(y),y)$. Let $z'$ denote the next unmarked vertex of $bd_c(z,p_v(z))$ in clockwise order from $z$ such that $z'$ is 
not visible from either $p_u(z)$ or $p_v(z)$. 
Note that, depending on the current vertex $z$, $z'$ may or may not exist. However, one of the following four mutually exclusive scenarios must be true.
\setlist{noitemsep}			
\begin{enumerate}[label=(\emph{\Alph*})]
 \item Every vertex of $bd_c(z,p_v(z))$ is already marked due to guards currently included in $S_B$ (see Figure \ref{locate_w1}).
 \item Every unmarked vertex of $bd_c(z,p_v(z))$ is visible from $p_u(z)$ or $p_v(z)$ (see Figure \ref{locate_w2}).
 \item Not every unmarked vertex of $bd_c(p_u(z'),z')$ is visible from $p_u(z')$ or $p_v(z')$ (see Figure \ref{locate_w3}). 
 \item Every unmarked vertex of $bd_c(p_u(z'),z')$ is visible from $p_u(z')$ or $p_v(z')$ (see Figure \ref{locate_w4}).
\end{enumerate}

If $z$ satisfies property (A) or (B), then $z$ is included in $B$ and the first unmarked vertex of $bd_c(p_v(z),v)$ in clockwise order from $p_v(z)$ 
becomes the new $z$ (see lines \ref{VG_4approx:6} to \ref{VG_4approx:9} of Algorithm \ref{VG_4approx}). If $z$ satisfies property (C), then $z$ is 
included in $B$ and $z'$ becomes the new $z$. If $z$ satisfies property (D), then $z'$ becomes the new $z$  (see lines \ref{VG_4approx:11} to \ref{VG_4approx:14} 
of Algorithm \ref{VG_4approx}). Whenever $z$ is included in $B$, $p_u(z)$ and $p_v(z)$ are included in $S_B$ and all unmarked vertices that become 
visible from $p_u(z)$ or $p_v(z)$ are marked. After doing so, if there remain unmarked vertices on $bd_{cc}(z,u)$, then $bd_{cc}(z,u)$ is scanned 
from $z$ in counterclockwise order and more guards are included in $S_B$ according to the following strategy (see lines \ref{VG_4approx:15} to 
\ref{VG_4approx:20} of Algorithm \ref{VG_4approx}).                          
\begin{enumerate}[label=(\emph{\roman*})]
 \item $y \leftarrow p_u(z)$
 \item Scan $bd_{cc}(y,u)$ from $y$ in counterclockwise till an unmarked vertex $x$ is located. 
 \item Add $x$ to $B$. Add $p_u(x)$ and $p_v(x)$ to $S_B$.
 \item Mark every vertex visible from $p_u(x)$ or $p_v(x)$.
 \item $y \leftarrow p_u(x)$
 \item Repeat steps (ii)-(v) until all vertices of $bd_{cc}(z,u)$ are marked. 
\end{enumerate} 

Initially, $z$ is chosen to be $u$ itself (see line \ref{VG_4approx:3} of Algorithm \ref{VG_4approx}). Then, for each $z$ under consideration along 
the clockwise scan of $bd_c(u,v)$, the appropriate action is performed corresponding to the property of $z$. Then, $z$ is updated and the process 
is repeated until $v$ is reached. The set of vertices $S_B$ is returned by the algorithm (see line \ref{VG_4approx:23} of Algorithm \ref{VG_4approx}) 
as a guard set. The entire process is described in pseudocode below as Algorithm \ref{VG_4approx}. 


\begin{algorithm}[H]
\caption{An $\mathcal{O}(n^2)$-algorithm for computing a guard set $S$ for all vertices of $P$}
\label{VG_4approx}
\begin{algorithmic}[1]
\State Compute $SPT(u)$ and $SPT(v)$ \label{VG_4approx:1} 
\State Initialize all the vertices of $P$ as unmarked \label{VG_4approx:2}
\State Initialize $B \leftarrow \emptyset$, $S_B \leftarrow \emptyset$ and $z \leftarrow u$  \label{VG_4approx:3}

\While { there exists an unmarked vertex in $P$ } \label{VG_4approx:4}
\State $z \leftarrow$ the first unmarked vertex on $bd_c(u,v)$ in clockwise order from $z$ \label{VG_4approx:5}

\If{ every unmarked vertex of $bd_c(z,p_v(z))$ is visible from $p_u(z)$ or $p_v(z)$ } \label{VG_4approx:6}
\State $B \leftarrow B \cup \{z\}$ and $S_B \leftarrow S_B \cup \{p_u(z),p_v(z)\}$ \label{VG_4approx:7}
\State Mark all vertices of $P$ that become visible from $p_u(z)$ or $p_v(z)$ \label{VG_4approx:8}
\State $z \leftarrow p_v(z)$ \label{VG_4approx:9}

\Else \label{VG_4approx:10}				
\State $z' \leftarrow$ the first unmarked vertex on $bd_c(z,v)$ in clockwise order \label{VG_4approx:11} 
\While { every unmarked vertex of $bd_c(p_u(z'),z')$ is visible from $p_u(z')$ or $p_v(z')$ } \label{VG_4approx:12}
\State $z \leftarrow z'$ and $z' \leftarrow$ the first unmarked vertex on $bd_c(z',v)$ in clockwise order \label{VG_4approx:13} 
\EndWhile \label{VG_4approx:14}

\State $w \leftarrow z$ \label{VG_4approx:15}                                    
\While { there exists an unmarked vertex on $bd_c(u,z)$ } \label{VG_4approx:16}
\State $B \leftarrow B \cup \{w\}$ and $S_B \leftarrow S_B \cup \{p_u(w),p_v(w)\}$ \label{VG_4approx:17}
\State Mark all vertices of $P$ that become visible from $p_u(w)$ or $p_v(w)$ \label{VG_4approx:18}
\State $w \leftarrow$ the first unmarked vertex on $bd_{cc}(w,u)$ in counterclockwise order \label{VG_4approx:19} 
\EndWhile \label{VG_4approx:20}

\EndIf \label{VG_4approx:21}
\EndWhile \label{VG_4approx:22}

\State \Return the guard set $S = S_B$ \label{VG_4approx:23}                    
\end{algorithmic}
\end{algorithm} 


\begin{figure}[H]
\begin{minipage}{.39\textwidth}
  \centerline{\includegraphics[width=.77\textwidth]{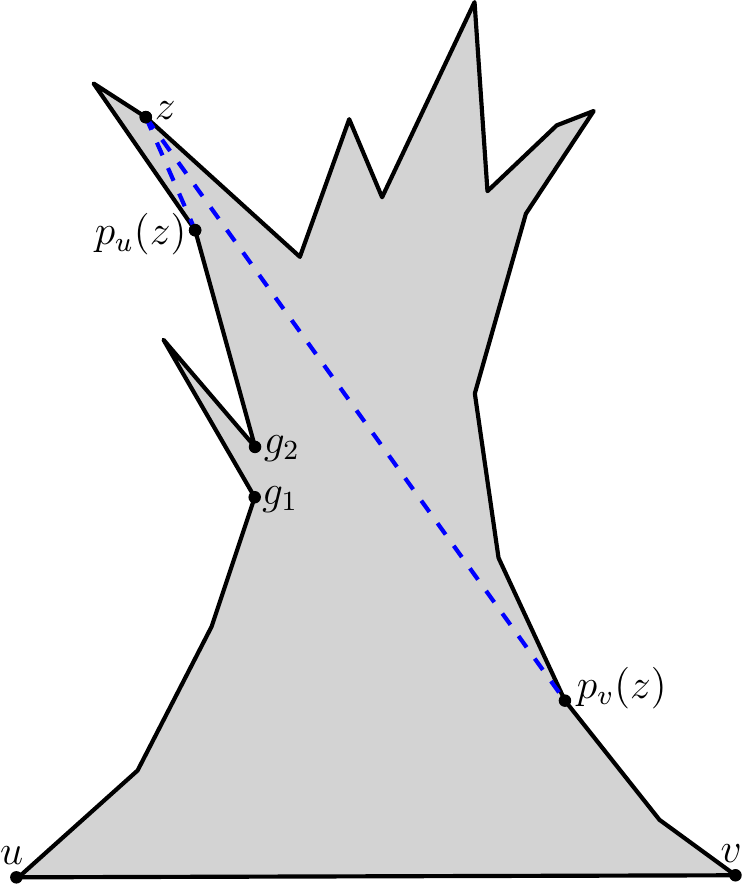}}
  \caption{All vertices of $bd_c(z,p_v(z))$ are already marked due to guards at $g_1$ \& $g_2$.} 
  \label{locate_w1}
\end{minipage}
\hspace*{19mm}
\begin{minipage}{.39\textwidth}
  \centerline{\includegraphics[width=.77\textwidth]{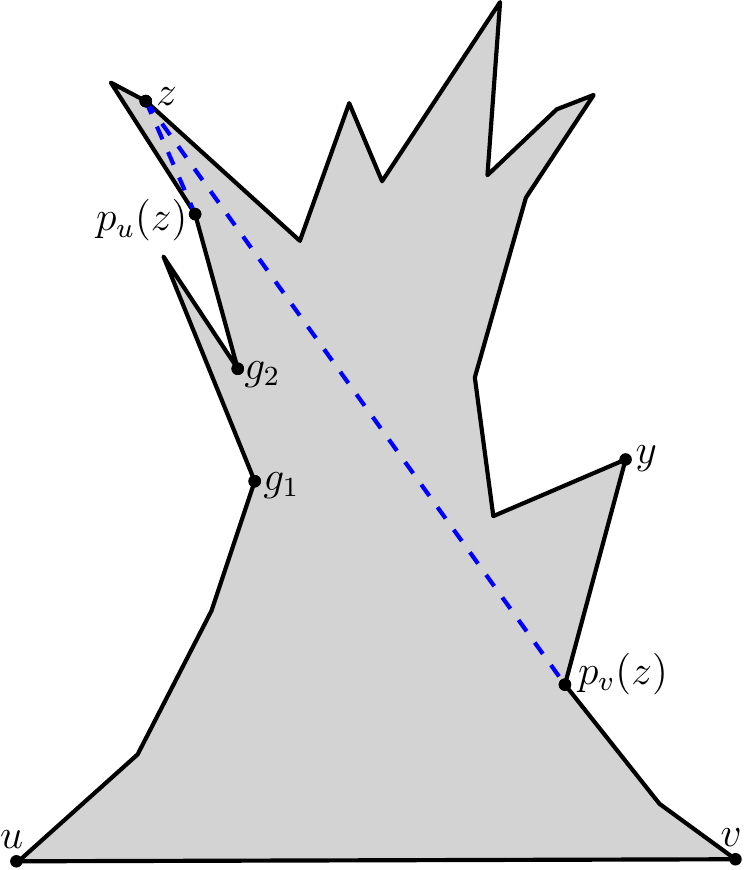}}
  \caption{The only unmarked vertex $y$ of $bd_c(z,p_v(z))$ is visible from $p_v(z)$.}
  \label{locate_w2}
\end{minipage}

\begin{minipage}{.39\textwidth}
  \centerline{\includegraphics[width=.77\textwidth]{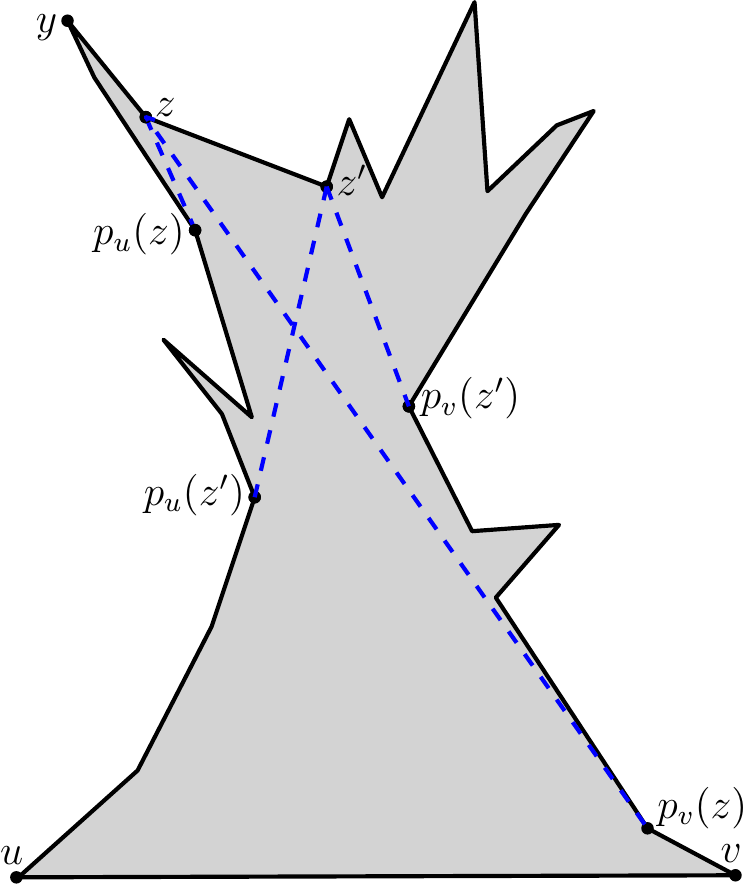}}
  \caption{Guards at $p_u(z')$ and $p_v(z')$ do not see the unmarked vertex $y$ of $bd_c(p_u(z'),z')$.} 
  \label{locate_w3}
\end{minipage}
\hspace*{19mm}
\begin{minipage}{.39\textwidth}
  \centerline{\includegraphics[width=.77\textwidth]{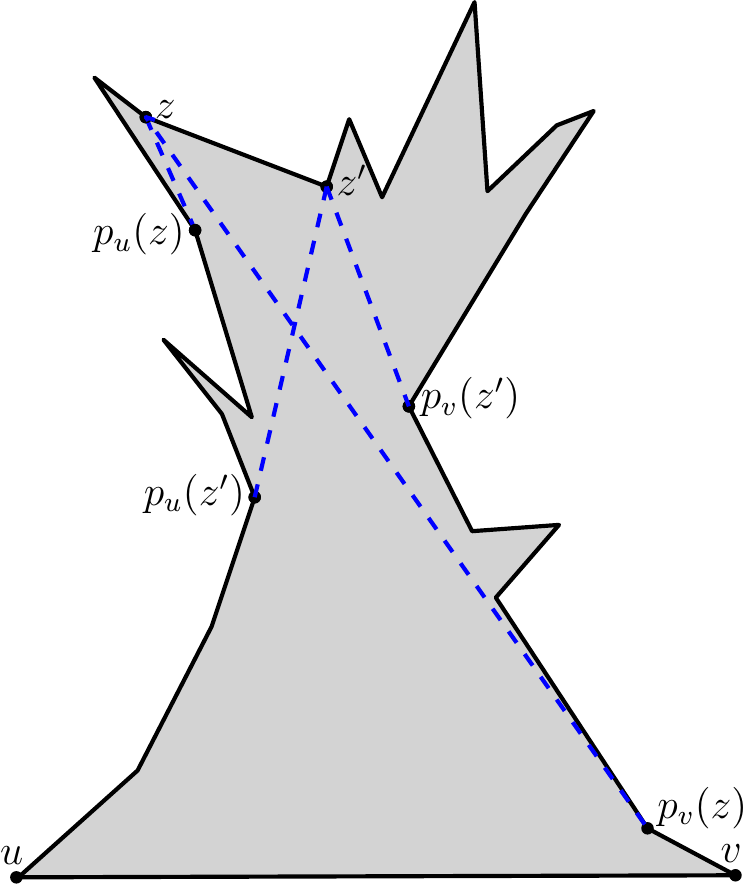}}
  \caption{Guards at $p_u(z')$ and $p_v(z')$ see all unmarked vertices of $bd_c(p_u(z'),z')$.}
  \label{locate_w4}
\end{minipage}
\end{figure}

Let us try to show an upper bound on $S$, by constructing a bipartite graph $G=(B \cup S_{opt},E)$ such that the degree of each vertex in $B$ is exactly 1 
and the degree of each vertex in $S_{opt}$ is at most 2. Let us denote by $b_i$ the $i$th vertex included in $B$ during the runtime of the algorithm. Now, 
each guard $S_{opt}$ that sees $b_i$ must be a vertex of $bd_c(p_u(b_i),p_v(b_i))$. We construct the graph $G$ by initially choosing, for each $b_i \in B$, 
a guard $g \in S_{opt}$ that sees $b_i$ and adding an edge $gb_i$ to $E$, which immediately implies that the degree of each vertex in $G$ belonging 
to $B$ is exactly 1. Note that, a single guard $g \in S_{opt}$ may see multiple vertices of $B$, and it may therefore have degree greater than 1 in $G$. 
By carefully reviewing some of the associations between guards in $S_{opt}$ and vertices in $B$, and making some adjustments to the set of edges $E$,
let us attempt to restrict to at most 2 the degree of each vertex in $G$ that belongs to $S_{opt}$. \\ 

In order to enforce this degree restriction, let us consider a guard $g \in S_{opt}$ that sees three distinct vertices $b_i, b_j, b_k \in B$, where $i < j < k$ 
and for any $l$ such that $i < l < j$ or $j < l < k$, vertex $b_l$ is not visible from $g$. Now, by Lemma \ref{l2}, $b_j$ or $b_k$ cannnot lie on 
$bd_{cc}(p_u(b_i),p_v(b_i))$, since any vertex visible from $g$ that lies on $bd_{cc}(p_u(b_i),p_v(b_i))$ is marked by Algorithm \ref{VG_4approx} 
when vertex $b_i$ is included in $B$. Also, due to the invariance maintained by Algorithm \ref{VG_4approx}, every unmarked vertex of $bd_c(p_u(b_i),b_i)$ 
is visible from $p_u(b_i)$ or $p_v(b_i)$ when $b_i$ is first considered as the current vertex, and is therefore marked by Algorithm \ref{VG_4approx} 
when vertex $b_i$ is included in $B$. So, $b_j$ or $b_k$ cannnot lie on $bd_c(p_u(b_i),b_i)$. Thus, both $b_j$ and $b_k$ must lie on $bd_c(b_i,p_v(b_i))$.  \\
 
Suppose the vertex $b_i$ is included in $B$ because it satisfies property (A) or (B), that is every unmarked vertex of $bd_c(b_i,p_v(b_i))$ is visible from 
$p_u(b_i)$ or $p_v(b_i)$, when it is considered to be the current vertex by Algorithm \ref{VG_4approx}. Then, the vertices $b_j$ and $b_k$ cannot exist. 
So, it must be the case that vertex $b_i$ satisfies property (C) or (D). Let us consider these two cases separately. \\

If the vertex $b_i$ satisfies property (C), that is, if we consider the next unmarked vertex $b'_i$ in clockwise order, not every unmarked vertex lying on 
$bd_c(p_u(b'_i),b'_i)$ is visible from $p_u(b'_i)$ or $p_v(b'_i)$. Since there do not exist any unmarked vertices on $bd_c(b_i,b'_i)$, it must be the case that 
$p_u(b'_i)$ lies on $bd_c(u,p_u(b_i))$ and there exists a vertex $x_i$ lying on $bd_c(p_u(b_i),b_i)$ such that $x_i$ is not visible from $p_u(b'_i)$ or $p_v(b'_i)$.
As $x_i$ is not visible from $p_v(b'_i)$, $x_i$ is not visible from any vertex that lies on $bd_c(b'_i,p_v(b'_i))$. Now, let us consider separately the following two 
subcases -- (i) $b_j$ and $b'_i$ are the same vertex, or $p_v(b_j)$ lies on $bd_c(b_j,p_v(b'_i))$; and, (ii) $p_v(b'_i)$ lies on $bd_c(b'_i,p_v(b_j))$. \\

If $b_j$ and $b'_i$ are the same vertex or $p_v(b_j)$ lies on $bd_c(b_j,p_v(b'_i))$, then $x_i$ is not visible from any vertex that lies on 
$bd_c(b_j,p_v(b_j))$. So, if we consider any guard $g' \in S_{opt}$ that sees $x_i$, $g'$ cannot lie on $bd_c(b_j,p_v(b_j))$. Note that the inclusion of $b_j$ in 
$B$ implies that $b_j$ is not visible from $p_u(b_i)$ or $p_v(b_i)$. Let $q_u$ be the vertex closest to $b_j$ on the Euclidean shortest path from $p_u(b_i)$ to 
$b_j$. Since $p_u(b_i)$ must lie on $bd_c(u,p_u(x_i))$, if $g'$ lies on $bd_c(p_u(b_j),q_u)$, then $g'$ cannot see $b_j$. Also, $g'$ cannot lie on $bd_c(q_u,b_j)$, 
since no vertex on $bd_c(q_u,b_j)$ is visible from $x_i$. Hence, any guard $g' \in S_{opt}$ which sees $x_i$ must lie outside $bd_c(p_u(b_j),p_v(b_j))$ and therefore 
be distinct from $g$. So, in this case, we delete the edge $gb_i$ in $G$ and insert the edge $g'b_i$ instead, thereby restricting the degree of $g$ in $G$ to 2. \\

\begin{figure}[H]
\centerline{\includegraphics[height=82mm]{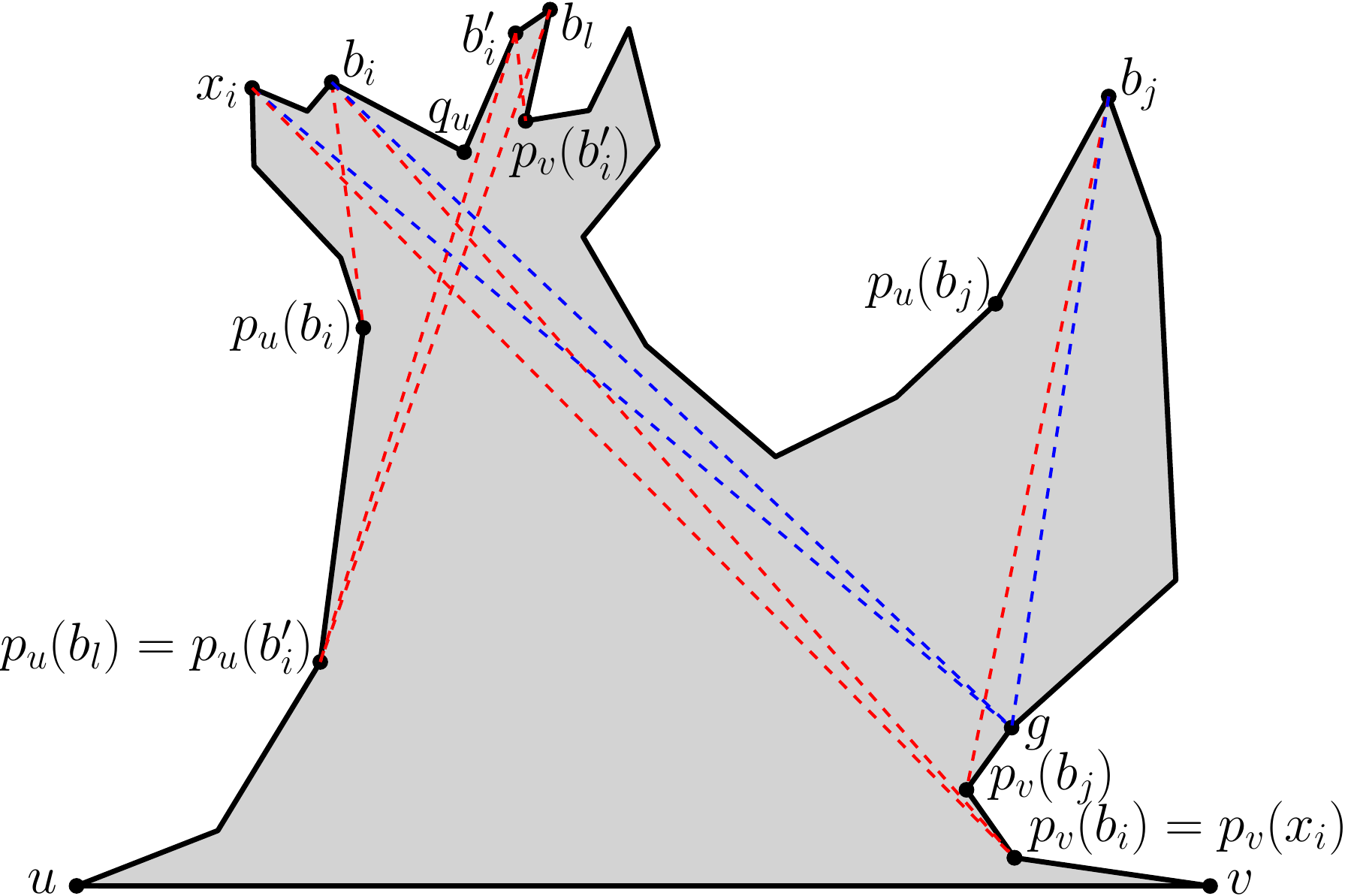}}
\caption{A possible situation where $p_v(b'_i)$ lies on $bd_c(b'_i,p_v(b_j))$.}
\label{b'i_bj}
\end{figure}

If $p_v(b'_i)$ lies on $bd_c(b'_i,p_v(b_j))$ (see Figure \ref{b'i_bj}), then there must exist a vertex $b_l \in B$ such that $i < l < j$ and $b_l$ lies on $bd_c(b_i,b_j)$. 
So, by our initial assumption, any guard $g'' \in S_{opt}$ that sees $b_l$ must be distinct from $g$. So, in this case, we delete the edge $gb_i$ in $G$ and insert 
the edge $g''b_i$ instead, thereby restricting the degree of $g$ in $G$ to 2. \\

If the vertex $b_i$ satisfies property (D), that is every unmarked vertex lying on $bd_c(p_u(b'_i),b'_i)$ is visible from $p_u(b'_i)$ or $p_v(b'_i)$ when it is first 
considered to be the current vertex by Algorithm \ref{VG_4approx}, then $b'_i$ is skipped initially and later included in $B$ when the algorithm backtracks to place 
guards for unmarked vertices lying on $bd_{cc}(p_u(b_{i-1}),u)$. Again, just like $b_i$, $b_j$ cannot be included in $B$ because it satisfies property (A) or (B), 
since the existence of $b_k$ leads to a contradiction from Lemma \ref{l2}. Now, in case that vertex $b_j$ is included in $B$ because it satisfies property (C), we can 
argue just as before that there exists a vertex $x_j$ lying on $bd_c(p_u(b_j),b_j)$ such that $x_j$ is not visible from $p_u(b'_j)$ or $p_v(b'_j)$, where $b'_j$ is the 
next unmarked vertex in clockwise order. Moreover, it follows that there must exist some other guard $g' \in S_{opt}$ distinct from $g$. So, in this case, we delete 
the edge $gb_j$ in $G$ and insert the edge $g'b_j$ instead, thereby restricting the degree of $g$ in $G$ to 2. However, a problem arises when $b_j$ also satisfies 
property (D), because then we cannot find some other guard in $S_{opt}$ distinct from $g$ with which we can associate it. In fact, note that we may have an 
arbitrarily long chain of vertices, all belonging to $B$, but satisfying property (D), which can jeopardize our attempts to restrict the degree of the single guard 
$g \in S_{opt}$ that sees all of them. \\

In order to prevent the above situation from happening, we modify our algorithm slightly. In the new algorithm, we maintain in a separate set $B'$ all the vertices that are 
included during backtracking. At the end of the clockwise scan, when all vertices have been marked, we check for redundant vertices in $B'$. A vertex $q$ is considered to 
be redundant and removed from the set $B'$ if every vertex that is marked due to the guards placed at $p_u(q)$ and $p_v(q)$ during its inclusion is also visible from the 
parents of some other vertex included later in $B'$. Therefore, the new algorithm implements this by running a backward scan over the vertices included in $B'$, in reverse 
order of inclusion, and marking every unmarked vertex visible from the parents of the current vertex under consideration. A particular vertex is eliminated during the scan if 
no new vertices are marked when it is considered as the current vertex. The modified algorithm is described in pseudocode below as Algorithm \ref{VG_4approx_patched}. \\

\begin{algorithm}[H]
\caption{An $\mathcal{O}(n^2)$-algorithm for computing a guard set $S$ for all vertices of $P$}
\label{VG_4approx_patched}
\begin{algorithmic}[1]
\State Compute $SPT(u)$ and $SPT(v)$ \label{VG_4approx_patched:1} 
\State Initialize all the vertices of $P$ as unmarked \label{VG_4approx_patched:2}
\State Initialize $B \leftarrow \emptyset$, $S_B \leftarrow \emptyset$, $B' \leftarrow \emptyset$, $S'_B \leftarrow \emptyset$ and $z \leftarrow u$ \label{VG_4approx_patched:3} 

\While { there exists an unmarked vertex in $P$ } \label{VG_4approx_patched:4}
\State $z \leftarrow$ the first unmarked vertex on $bd_c(u,v)$ in clockwise order from $z$ \label{VG_4approx_patched:5}
\If{ every unmarked vertex of $bd_c(z,p_v(z))$ is visible from $p_u(z)$ or $p_v(z)$ } \label{VG_4approx_patched:6}
\State $B \leftarrow B \cup \{z\}$ and $S_B \leftarrow S_B \cup \{p_u(z),p_v(z)\}$ \label{VG_4approx_patched:7}
\State Mark all vertices of $P$ that become visible from $p_u(z)$ or $p_v(z)$ \label{VG_4approx_patched:8}
\State $z \leftarrow p_v(z)$ \label{VG_4approx_patched:9}
\Else \label{VG_4approx_patched:10}				
\State $z' \leftarrow$ the first unmarked vertex on $bd_c(z,v)$ in clockwise order \label{VG_4approx_patched:11} 
\While { every unmarked vertex of $bd_c(p_u(z'),z')$ is visible from $p_u(z')$ or $p_v(z')$ } \label{VG_4approx_patched:12}
\State $z \leftarrow z'$ and $z' \leftarrow$ the first unmarked vertex on $bd_c(z',v)$ in clockwise order \label{VG_4approx_patched:13} 
\EndWhile \label{VG_4approx_patched:14}    
\State $B \leftarrow B \cup \{z\}$ and $S_B \leftarrow S_B \cup \{p_u(z),p_v(z)\}$ \label{VG_4approx_patched:15}            
\While { there exists an unmarked vertex on $bd_c(u,z)$ } \label{VG_4approx_patched:16}
\State $w \leftarrow$ the first unmarked vertex on $bd_{cc}(z,u)$ in counterclockwise order \label{VG_4approx_patched:17} 
\State $B' \leftarrow B' \cup \{w\}$ and $S'_B \leftarrow S'_B \cup \{p_u(w),p_v(w)\}$ \label{VG_4approx_patched:18}
\State Mark all vertices of $P$ that become visible from $p_u(w)$ or $p_v(w)$ \label{VG_4approx_patched:19}
\EndWhile \label{VG_4approx_patched:20}
\EndIf \label{VG_4approx_patched:21}
\EndWhile \label{VG_4approx_patched:22}
\State Reinitialize all the vertices of $P$ that are visible from some guard in $S_B$ as unmarked \label{VG_4approx_patched:23}
\For{ each vertex $z \in B'$ chosen in reverse order of inclusion } \label{VG_4approx_patched:24}
\State Locate and mark each unmarked vertex visible from $p_u(z)$ or $p_v(z)$ \label{VG_4approx_patched:25}
\If{ no new vertices get marked due to guards at $p_u(z)$ or $p_v(z)$ } \label{VG_4approx_patched:26}
\State $B' \leftarrow B' \setminus \{z\}$  and $S'_B \leftarrow S'_B \setminus \{p_u(w),p_v(w)\}$ \label{VG_4approx_patched:27}

\EndIf \label{VG_4approx_patched:28}
\EndFor \label{VG_4approx_patched:29}

\State $B \leftarrow B \cup B'$ \label{VG_4approx_patched:30}
\State \Return the guard set $S = S_B \cup S'_B$ \label{VG_4approx_patched:31}                   
\end{algorithmic}
\end{algorithm} 

Observe that Algorithm \ref{VG_4approx_patched} eliminates from the set $B$ precisely those vertices which we previously found impossible 
to reassociate with a different guard in $S_{opt}$, in case the initial guard with which we associated it already had edges in the bipartite graph 
$G$ incident on it from more than two vertices of $B$. So, if we now revisit our strategy for constructing the bipartite graph $G$ in order to 
associate guards in $S_{opt}$ with guards in $B$, as computed by Algorithm \ref{VG_4approx_patched}, the following lemma must be true.

\begin{lemma} \label{l6}
In the bipartite graph $G$, the degree of each vertex in $B$ is exactly 1 and degree of each vertex in $S_{opt}$ is at most 2.
\end{lemma}


\begin{corollary} \label{l7}
$|B| \leq 2|S_{opt}|$.  
\end{corollary}

\begin{theorem} \label{l8}
$|S| \leq 4|S_{opt}|$. 
\end{theorem}  
\begin{proof}
By arguments similar to those in the proof of Lemma \ref{l4}, $|S_B| = 2|B|$. Also, by Corollary \ref{l7}, $|B| \leq 2|S_{opt}|$.
Therefore, $|S| = |S_B| = 2|B| \leq 4|S_{opt}|$. 
\end{proof}


\subsection{Guarding all interior points of a polygon}
\label{algo2}
In the previous subsection, we presented an algorithm (see Algorithm \ref{VG_4approx}) which returns a guard set $S$ such that all vertices of $P$ 
are visible from guards in $S$.  However, it may not always be true that all interior points of $P$ are also visible from guards in $S$. Consider 
the polygon shown in Figure \ref{pocket_edge}. While scanning $bd_c(u,v)$, our algorithm places guards at $p_u(z)$ and $p_v(z)$ as all vertices 
of $bd_c(p_u(z),p_v(z))$ become visible from $p_u(z)$ or $p_v(z)$. Observe that in fact all vertices of $P$ become visible from these two guards. 
However, the triangular region $P \setminus (VP(p_u(z)) \cup VP(p_v(z)))$, bounded by the segments $x_1 x_2$, $x_2 x_3$ and $x_3 x_1$, is not visible 
from $p_u(z)$ or $p_v(z)$. Also, one of the sides $x_1 x_2$ of the triangle $x_1 x_2 x_3$ is a part of the polygonal edge $a_1 a_2$. In fact, for any 
such region invisible from guards in $S$, one of the sides must always be a part of a polygonal edge. 
Otherwise, there should exist another guard $g$ (see Figure \ref{pocket_edge}) from which the entire polygonal side ($x_1 x_2$) of the region is 
visible and yet some portion of the region (including $x_3$) is not visible. However, such a vertex $g$ cannot be weakly visible from the edge $uv$,
which is a contradiction. Henceforth, any such region invisible from guards in $S$ is referred to as an \emph{invisible cell}, and the polygonal edge 
which contributes as a side to the invisible cell is referred to as its corresponding \emph{partially invisible edge}. 
One additional guard is required in order to see each invisible cell entirely. For example, in Figure \ref{pocket_edge}, an extra guard is required 
at a vertex of $bd_c(z,w)$, since none of the vertices outside this boundary can see all points of the invisible cell $x_1 x_2 x_3$. \\


The boundary of the visibility polygon $VP(s)$ of any vertex $s$ consists of polygonal edges and constructed edges. A \emph{constructed edge} $yx$ 
is an edge formed by extending the segment $sy$ (which could be either an edge of $P$ or an internal segment), where $y$ is some other vertex of $P$,
till it touches the boundary of $P$ at a point $x$. If $y$ lies on $bd_c(s,x)$, the region of $P$ bounded by $bd_c(y,x)$ and $xy$ is referred to as 
the \emph{left pocket} of $VP(z)$. Similarly, if $y$ lies on $bd_{cc}(s,x)$, then the region of $P$ bounded by $bd_{cc}(y,x)$ and $xy$ is referred to 
as the \emph{right pocket} of $VP(z)$. In both these cases, we refer to the vertex $y$ as the \emph{lid vertex} and the point $x$ as the \emph{lid point} 
of the corresponding left or right pocket. 

\begin{figure}[H]
\begin{minipage}{.49\textwidth}
  \centerline{\includegraphics[width=.84\textwidth]{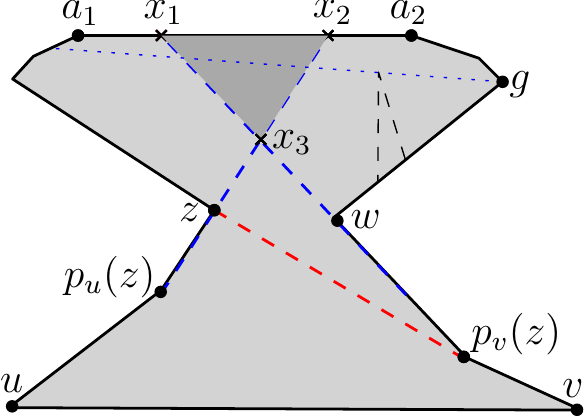}}
  \caption{All vertices are visible from $p_u(z)$ or $p_v(z)$, but the triangle $x_1 x_2 x_3$ is invisible.}
  \label{pocket_edge}
\end{minipage}
\hspace{1mm} 
\begin{minipage}{.49\textwidth}
  \centerline{\includegraphics[width=.84\textwidth]{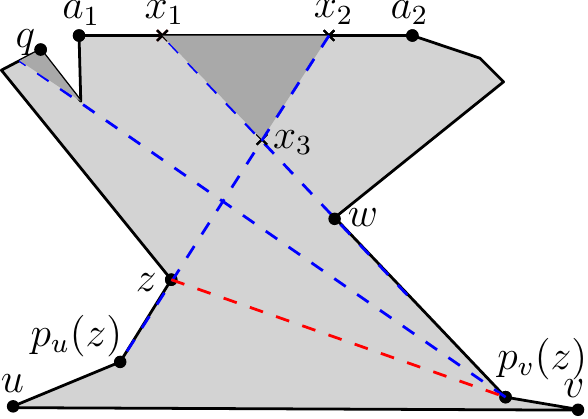}}
  \caption{The left pocket of $VP(p_u(z))$ can contain only one invisible cell.}
  \label{pocket_multiple} 
\end{minipage}
\end{figure}

Observe that each invisible cell must be wholly contained within the intersection region (which is a triangle) of a left pocket and a right pocket. 
For example, in Figure \ref{pocket_edge}, the invisible cell $x_1 x_2 x_3$ is actually the entire intersection region of the left pocket of $VP(p_u(z))$ 
and the right pocket of $VP(p_v(z))$. Also, $z$ is the lid vertex and $x_2$ is the lid point of the left pocket of $VP(p_u(z))$. Similarly, $w$ is the 
lid vertex and $x_1$ is the lid point of the right pocket of $VP(p_v(z))$. \\

Suppose $bd_c(z,x_2)$ contains reflex vertices (see Figure \ref{pocket_multiple}). In that case, in addition to the invisible cell $x_1 x_2 x_3$, 
the left pocket of $VP(p_u(z))$ may contain several regions that are not visible from $p_v(z)$. However, in each such region there exists a vertex, 
say $q$, that is not visible from $p_v(z)$, which contradicts the fact that all vertices of $bd_c(p_u(z),p_v(z))$ are visible from $p_u(z)$ or 
$p_v(z)$. So, the left pocket of $VP(p_u(z))$ can contain only one invisible cell. Analogously, the right pocket of $VP(p_v(z))$ can contain only 
one invisible cell. \\




Now consider the situation (as shown in Figure \ref{pockets}) where $VP(p_u(z))$ has several left pockets and $VP(p_v(z))$ has several right pockets 
which intersect pairwise to create multiple invisible cells. In order to guard these invisible cells, additional guards are placed as follows. Let 
$c_1$ be the lid point of the left pocket containing the first invisible cell in clockwise order. Then, guards are placed at $p_u(c_1)$ and $p_v(c_1)$. 
Now, for every invisible cell $T$, the portions of $T$ are removed that are visible from $p_u(c_1)$ or $p_v(c_1)$. Note that some of these cells may 
turn out to be totally visible and hence may be eliminated altogether. This process is repeated until all invisible cells become totally visible. \\
 
\begin{figure}[H]
\centerline{\includegraphics[width=\textwidth]{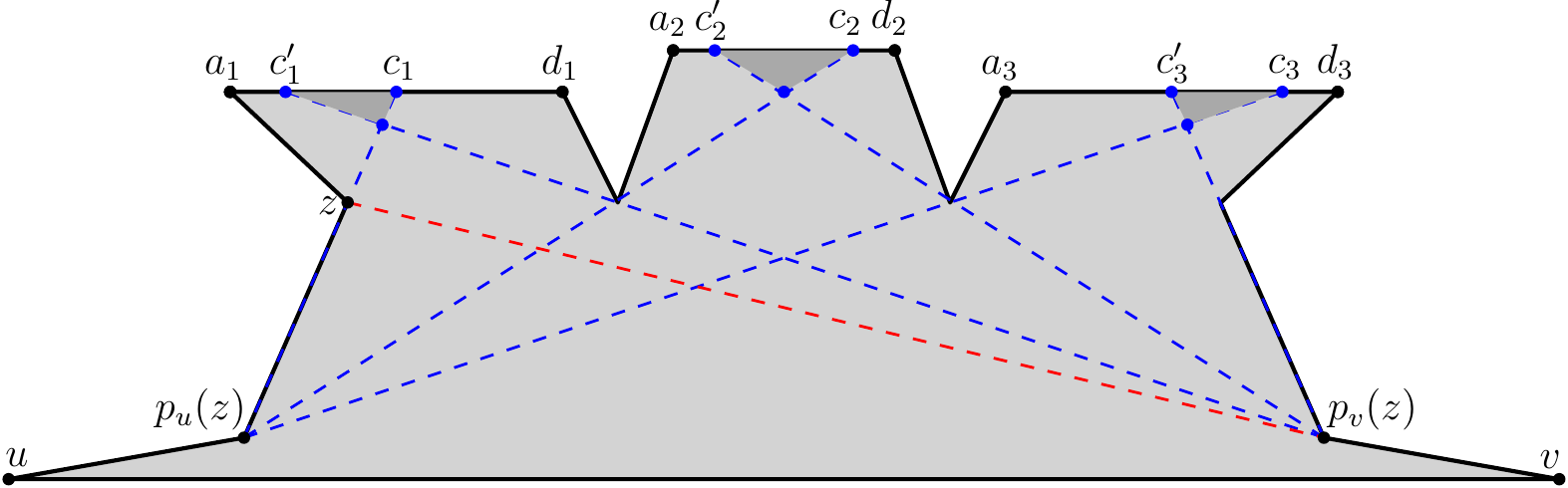}}
\caption{Multiple invisible cells exist within the polygon that are not visible from the guards placed at $p_u(z)$ and $p_v(z)$.}
\label{pockets}
\end{figure}

In general, we may have a situation where multiple invisible cells are created by the intersection of the left and right pockets of arbitrary pairs
of guards belonging to $S$ (see Figure \ref{theorem3}). In this scenario, all invisible cells are guarded by introducing a set of additional guards 
$S'$ as follows. Initially, both $C$ and $S'$ are empty. Scan $bd_c(u,v)$ from $u$ in clockwise order to locate the first edge $a_i d_i$ that is 
not totally visible from guards in $S \cup S'$, where $d_i$ is the next clockwise vertex of $a_i$. Let $c_i' c_i$ be the portion of $a_i d_i$ that is 
not visible from guards in $S \cup S'$, where $c_i' \in bd_c(a_i,c_i)$ and $c_i \in bd_c(c_i',d_i)$. In other words, $c_i' c_i$ is the polygonal side of 
the first invisible cell. Add $p_u(c_i)$ and $p_v(c_i)$ to $S'$. Also, add $c_i$ to $C$. Repeat this process until all the edges of $P$ are totally 
visible from guards in $S \cup S'$. At its termination, let us assume that $C = \{c_1,c_2,\dots,c_k\}$. The entire procedure is described in 
pseudocode as Algorithm \ref{VG_6approx}.

\begin{algorithm}
\caption{An $\mathcal{O}(n^2)$-algorithm for computing a guard set $S \cup S'$ for guarding $P$ entirely} 
\label{VG_6approx}
\begin{algorithmic}[1]
\State Compute $SPT(u)$ and $SPT(v)$ \label{VG_6approx:1} 
\State Compute the set of guards $S$ using Algorithm \ref{VG_4approx_patched} \label{VG_6approx:2}
\State Initialize $C \leftarrow \emptyset$, $S' \leftarrow \emptyset$ and $z \leftarrow u$ \label{VG_6approx:3} 

\While { there exists an edge in $P$ that is partially visible from guards in $S \cup S'$  } \label{VG_6approx:4}
\State $z' \leftarrow$ the vertex next to $z$ in clockwise order on on $bd_c(u,v)$ \label{VG_6approx:5}

\If{ if the edge $zz'$ is partially visible from guards in $S \cup S'$ } \label{VG_6approx:6}
\State $c_i \leftarrow$ the lid point of the left pocket on $zz'$ \label{VG_6approx:7}
\State $C \leftarrow C \cup \{c_i\}$ and $S' \leftarrow S' \cup \{p_u(c_i),p_v(c_i)\}$ \label{VG_6approx:8}
\EndIf \label{VG_6approx:9}

\State $z \leftarrow z'$ \label{VG_6approx:10}
\EndWhile \label{VG_6approx:11}

\State \Return the guard set $S \cup S'$ \label{VG_6approx:12}
\end{algorithmic}
\end{algorithm}

\begin{theorem} \label{t9}
The running time of Algorithm \ref{VG_6approx} is $\mathcal{O}(n^2)$.   
\end{theorem}

\begin{proof}
$SPT(u)$ and $SPT(v)$ can be computed in $\mathcal{O}(n)$ time \cite{GHLST_1987}. Then, the computation of the guard set $S$ 
takes $\mathcal{O}(n^2)$ time, since it involves scanning the boundary of $P$ and identifying vertices to be marked whenever new guards are placed. 
The number of lid points on an edge can be at most $\mathcal{O}(n)$. Therefore, each time a new vertex is added to $S'$, the invisible portion of 
the first partially visible edge in clockwise order can be determined in $\mathcal{O}(n)$ time. Hence, the overall running time of Algorithm 
\ref{VG_6approx} is $\mathcal{O}(n^2)$. 
\end{proof}

\pagebreak
We have the following lemma connecting $S'$ with $S_{opt}$.

\begin{lemma} \label{l10}
$2|C| = |S'| \leq 2|S_{opt}|$.   
\end{lemma}

\begin{proof}
For every $c_i \in C$, there exists an invisible cell $T_i$. For every such invisible cell $T_i$, let $l_i$ and $r_i$ respectively denote the 
lid vertices of the left and right pockets intersecting to form $T_i$ (see Figure \ref{theorem3}). Let $g \in S$ be the guard such that $l_i$ is the 
lid vertex of a left pocket of $VP(g)$. Similarly, let $g' \in S$ be the guard such that $r_i$ is the lid vertex of a right pocket of $VP(g')$. \\

Assume that, for every $T_i$, there exists at least one guard in $S_{opt}$ that sees all points of $T_i$. 
Now, consider any guard $g_{opt} \in S_{opt}$ that sees all points of $T_i$. Then, $g_{opt}$ can lie on $bd_c(l_i,r_i)$. Also, $g_{opt}$ can lie on 
$bd_c(p_u(c_i),g)$, but only when $p_u(c_i) \neq l_i$ and $p_u(c_i)$ lies on $bd_c(u,g)$. Now, let $z$ be the vertex such that $p_v(z) = g'$. Then,  
no vertex of $bd_c(z,g')$ is visible from any vertex of $bd_c(g',v)$. Further, if $z$ is such that $p_u(z) = g$, then $z$ has to lie on $bd_c(g,l_i)$. 
Otherwise, $z$ has to lie on $bd_c(l_i,c'_i)$. In either case, $g_{opt}$ cannot lie on $bd_c(g',v)$ since $c'_i$ lies on $bd_c(z,g')$. \\

Since the guard set $S'$ includes $p_u(z)$ and $p_v(z)$ for every $z \in C$, clearly $|S'|=2|C|$. If for every $i$, there exists an unique vertex 
belonging to $S_{opt}$ that sees all points of $T_i$, then obviously $|S'| \leq 2|S_{opt}|$. 
Consider the special situation where $l_{i+1} = r_i$ for some $i$ (see Figure \ref{pockets}) so that both $T_i$ and $T_{i+1}$ are totally visible from 
$r_i$. Since all points of $T_i$ are visible from $r_i$, it must be the case that $p_v(c_i) = r_i$. Moreover, $r_i$ can be a vertex of $S_{opt}$. 
Therefore, no additional guards are chosen for $T_{i+1}$ because all points of $T_{i+1}$ become visible from the guard already placed at $r_i$. 

\begin{figure}[H]
\centerline{\includegraphics[width=\textwidth]{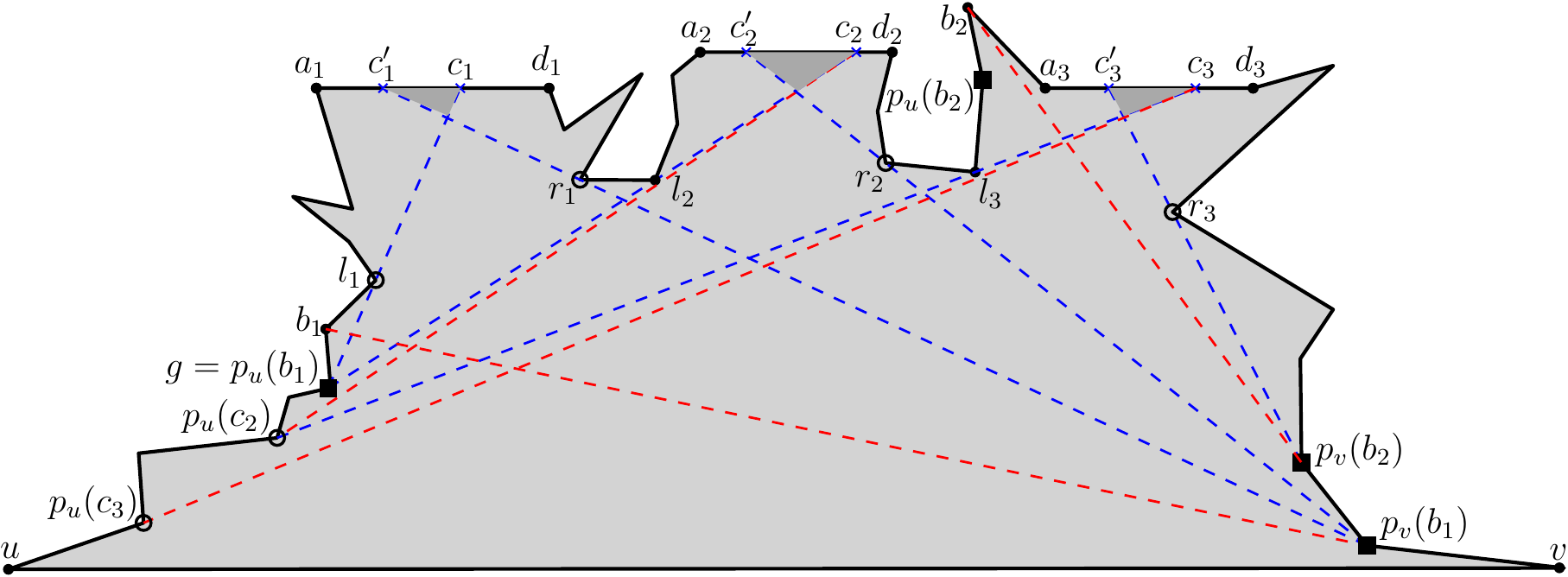}}
\caption{Placement of guards to in order to see all invisible cells.}
\label{theorem3}
\end{figure}

If no vertex of $bd_c(l_i,r_i)$ belongs to $S_{opt}$, then there must be a vertex of $S_{opt}$ lying on $bd_c(p_u(c_i),g)$ and $p_u(c_i)$ must belong 
to $bd_c(u,g)$. If $p_u(c_{i-1})$ also belongs to $bd_c(u,g)$, then $S_{opt}$ must have a vertex on the boundary $bd_c(p_u(c_i),p_v(c_{i-1}))$ in order 
to see $T_{i-1}$ because $l_{i-1}$ is the lid vertex of a left pocket of $VP(p_u(c_{i-1}))$. Hence, $2|C| = |S'| \leq 2|S_{opt}|$. \\

Finally, if we remove the assumption that there exists at least one guard in $S_{opt}$ that sees all points of $T_i$, then the size of $S_{opt}$ 
increases but our guard set $S'$ remains the same. Therefore, the bound is still preserved.
\end{proof}
\begin{theorem} \label{t4} 
$|S \cup S'| \leq 6|S_{opt}|$.   
\end{theorem}

\begin{proof}
By Lemma \ref{l5} and Lemma \ref{l10}, $|S \cup S'| \leq |S| + |S'| \leq 4|S_{opt}| + 2|S_{opt}| \leq 6|S_{opt}|$.  
\end{proof}


\section{Inapproximability of vertex guard problem in weak visibility polygons with holes}
\label{reduction}

Given a weak visibility polygon $P$ with holes, having $n$ vertices, the aim of the Vertex Guard problem is to find a smallest 
subset $S$ of the set of vertices of $P$ such that every point in the interior of the polygon $P$ can be seen from at least one 
vertex in $S$. The vertices in $S$ are called \emph{vertex guards}. In this section, we show an inapproximability result for the 
Vertex Guard problem in a weak visibility polygon with holes by showing how to construct an instance of Vertex Guard for every
instance of Set Cover. In Section \ref{reduction_existing}, we describe an existing reduction for general polygons with holes 
given by Eidenbenz, Stamm and Widmayer \cite{ESW_1998}. Then, in Section \ref{reduction_modified}, we modify this reduction so 
that it works even for polygons with holes that are weakly visible from an edge.

\subsection{Existing reduction for general polygons with holes}
\label{reduction_existing}

An instance of Set Cover consists of a finite universe $E = \{e_1.e_2,\dots,e_n\}$ of elements $e_j$ and a collection $S = \{s_1,s_2,\dots,s_m\}$ 
of subsets $s_i$ where each $s_i \subseteq E$. The problem is to find $S' \subseteq S$ of minimum cardinality such that every element $e_i$, for
$1 \leq i \leq n$, belongs to at least one subset in $S'$. For the ease of discussion, the elements in $E$ and the subsets in $S$ are assumed to 
have an arbitrary, but fixed order. \\ 


As shown in Figure \ref{existing}, a polygon is constructed in the $x-y$ plane. For every set $s_i$ ($1 \leq i \leq m$), a
point $((i-1)d',a)$ is placed on the horizontal line $y=a$ with a constant distance $d'$ between any two consecutive points.
For simplicity, the $i$th such point from the left is also referred to as $s_i$. Corresponding to every element $e_j \in E$,
two points $(D_j,0)$ and $(D_j',0)$ are placed on the horizontal line $y=0$, where $D_1 \geq 0$ and $D_j' = D_j + d$ for a 
positive constant $d$. The points are arranged from left to right, and for each $j=1,\dots,n$, they are referred to as $D_j$ 
and $D_j'$. For each $j=1,\dots,n$, the distance $d_j = D_{j+1} - D_j'$ is defined later. \\

Let $s_k$ and $s_l$ be respectively the first and last sets of which $e_j$ is a member. Without loss of generality, assume
that $s_k$ and $s_l$ are distinct. A line $g$ is drawn through $s_k$ and $D_j$. Also, a line $g'$ is drawn through $s_l$ and 
$D_j'$. Naming the intersection point of $g$ and $g'$ as $I_j$, the triangle $D_j I_j D'_j$ is called a \emph{spike}. Since it 
plays a crucial role in the construction, the point $I_j$ of each spike is called the \emph{distinguished point} of the spike. \\

For any pair $(i,j)$, if the set $s_i$ contains the element $e_j$, then two lines are drawn connecting $s_i$ with $D_j$ and 
$D'_j$, and the area between these two lines is called a \emph{cone}. Observe that, among all the lines mentioned so far, 
only the line segments of the horizontal line $y=0$ that are between adjacent spikes and the spikes themselves contribute edges 
to the polygonal boundary whereas all other lines just help in the construction. \\

The correspondence between an instance of Vertex Guard and an instance of Set Cover is established by ensuring that an optimal 
set of vertex guards includes only those points $s_i$ which belong to an optimal solution of Set Cover. So, in the construction,
a guard at vertex $s_i$  must see the spike of only those elements $e_j$ that are members of the set $s_i$. This is realized by 
introducing a \emph{barrier line} at $y=b$ such that only line segments on the horizontal line $y=b$ lying outside the cones are 
part of the polygonal boundary (see Figure \ref{existing}). Another barrier line at $y=b+b'$ is introduced at a distance of $b'$ 
from the first barrier. Holes of the polygon are defined by connecting each pair of points that is created by the intersection of 
the same cone-defining line with the barrier lines. The area between the two lines at $y=b$ and $y=b+b'$ is called the \emph{barrier}. 
Note that the barrier includes all the holes and it also contains a small part of every cone. 

\begin{figure}[H]
\centerline{\includegraphics[height=90mm]{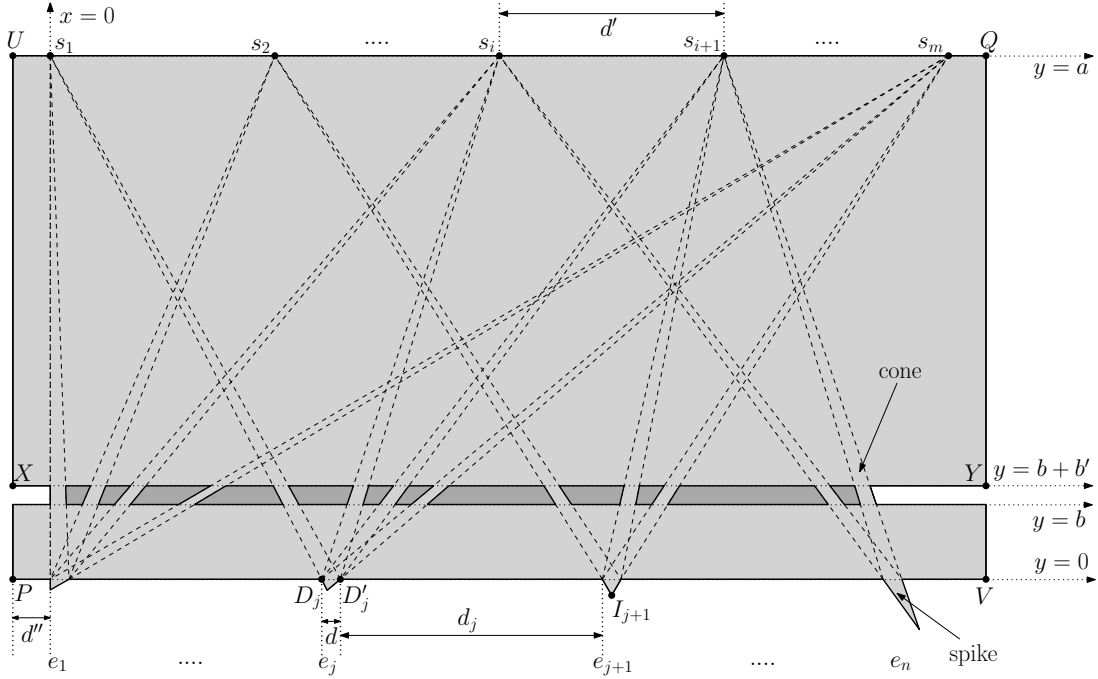}}
\caption{The existing reduction for general polygons with holes}
\label{existing}
\end{figure}

For every pair $(i,j)$, let us denote the point at $y=b$ on the line $s_i D_j$ as $w_{ij}$, and similarly, the point at $y=b$ on 
the line $s_i D'_j$ as $w'_{ij}$. Now, the thickness $b'$ of the barrier is to be determined in such a way that, for every hole, 
all segments of its boundary excluding those on the line $y=b+b'$ is visible from two guards at $P = (-d'',0)$ and $V = (D_n'+L,0)$. 
To achieve this, the thickness $b'$ is determined by intersecting, for each pair $(i,j)$, a line from $P$ through $w_{ij}$ and a 
line from $V$ through $w'_{ij}$. Then, $b'$ is assigned a value such that the barrier line $y = b+b'$ goes through the lowest of 
all these intersection points. \\

To complete the construction, a vertical line segment $PU$ at $x=-d''$ is drawn from $y=0$ to $y=y_0$, where $d''$ is a positive constant. 
Except for the portion of it between the two barrier lines, this line segment forms a part of the polygonal boundary. Also, a horizontal 
line segment is drawn from $D_n'$ to the point $V$ at $(D_n'+L,0)$. Finally, a point $Q$ is located at $(D_n'+L,a)$ and the external 
boundary of the polygon is completed by drawing the line segments $UQ$ and $QV$, except for the portion of $QV$ lying between the barrier lines. 
The points on the segments $PU$ and $QV$ that lie on the barrier line $y=b+b'$ are referred to as $X$ and $Y$ respectively, \\

Let $d'$, $d''$ and $a$ be arbitrary positive constants. The rest of the parameters are set in terms of $d'$, $d''$ and $a$ as follows: 
$d = \frac{d'}{4}$, $b = \frac{5}{12}a$, 
$b' = \frac{\frac{35}{144}a}{-4^{l-1}m^{l-1}+2\sum_{i=0}^{l-1}4^im^i+2\frac{d''}{d}-\frac{19}{12}}$,
and $D_l = -4^{l-1}m^{l-1}d-d+2d\sum_{i=0}^{l-1}4^im^i$ for $l=1,\dots,n$. 
As a consequence of these parameter settings, the following properties hold for this reduction.
\begin{itemize}
\item No three cones connecting different sets with different elements can overlap.      
\item The barrier is such that:
\begin{enumerate}[label=(\emph{\alph*})]
 \item All the intersections of cones from the same element $e_j$ are below $y=b$. 
 \item All intersections of cones from different elements are above $y=b+b'$.
 \item All of the barrier is visible from at least one of the two guards at $P$ and $V$, except for the line segments at $y=b+b'$.
\end{enumerate}
 \item The spikes of no two elements intersect.
\end{itemize}


\subsection{Modified reduction for weak visibility polygons with holes}
\label{reduction_modified}
To incorporate weak visibility from an edge, the known construction from Section \ref{reduction_existing} is modified as follows. \\ 

Let $R$ be the set of all rays $\overrightarrow{D_j s_i}$ and $\overrightarrow{D'_j s_i}$ such that the spike corresponding to $e_j$ is    
visible from $s_i$. For every pair $(i,j)$, the point of intersection of the ray $\overrightarrow{D_j s_i}$ with the barrier line $y=b+b'$ 
is denoted as $y_{i,j}$ (see Figure \ref{m+d}). Let $R'$ be the set of all rays $\overrightarrow{I_j y_{i,j}}$ such that the spike 
corresponding to $e_j$ is visible from $s_i$. Let $\alpha$ be the largest among all the angles made by rays belonging to $R \cup R'$ with 
the positive X-axis at $y=0$. A line $l'$ is constructed such that $l'$ passes through $s_m$ and makes an angle $\theta = \alpha + \frac{180-\alpha}{2}$ 
with the positive X-axis at $y=0$. The line $l'$ is translated to obtain another line $l$ in such a way that all holes contained within the 
barrier lie below $l$. The point of intersection of $l$ with the line $y=0$ is called $V$, whereas the point of intersection of the segment 
$PU$ with the barrier line $y=b+b'$ is called $X$. Also, the top right vertex of the rightmost hole contained within the barrier is referred 
to as $Y$. \\

Let $\beta$ be the maximum among all the angles made by the rays $\overrightarrow{Y s_i}$ with the positive X-axis at $y=a$. Among all points 
of intersection of $l$ with various rays belonging to $R \cup R'$, let $U'$ be the leftmost point. Then, a point $U=(x_u,y_u)$ is located along 
the ray $VU'$ such that, for every $i$, the angle made by the ray $\overrightarrow{U s_i}$ with the positive X-axis at $y=a$ is greater than 
$\beta$ (not represented accurately in Figure \ref{m+d} due to space constraints). Then, the external boundary of the polygon is completed by 
drawing the segments $PU$, $PV$ and $UV$, except for the portion of $PU$ lying between the barrier lines. The modified construction ensures that 
all spikes are totally visible from the edge $UV$. However, no distinguished point is visible from the point $U$ itself (see Figure \ref{m+d}). \\ 


Let $S_U$ and $S_V$ denote the set of all rays of the form $\overrightarrow{s_i U}$ and $\overrightarrow{Y s_i}$ respectively. Corresponding to 
every set $s_i$, let $S_i$ be the set of all rays $\overrightarrow{D_j s_i}$ and $\overrightarrow{D'_j s_i}$ such that the spike corresponding 
to $e_j$ is visible from $s_i$. Now, let $S = S_1 \cup S_2 \cup \dots \cup S_m$. Also, let $Z$ be the set of all points of intersection between
any two rays belonging to the set $S \cup S_U \cup S_Y$ that lie above the horizontal line $y=a$ passing through every $s_i$. Now, a horizontal 
line $y=a+a'$ is chosen such that it lies below all the points belonging to $Z$. For every $s_i$, a clockwise angular scan is performed around 
$s_i$ starting from the angle defined by $\overrightarrow{s_i U}$ till an angular region is located that is contained in no cone. Two rays 
$\overrightarrow{r_i}$ and $\overrightarrow{r'_i}$ are drawn within this region such that they intersect the line $y=a+a'$ at $z_i$ and $z'_i$
respectively. Then, corresponding to each $s_i$, a triangular hole is created by joining the segments $s_i z_i$, $s_i z'_i$ and $z_i z'_i$ 
(see Figure \ref{m+d}). Note that the entire region of the constructed polygon lying above the line $y=b+b'$ is weakly visible from the 
edge $UV$. Moreover, this entire region is also visible from two guards placed at $U$ and $Y$. 


\begin{lemma} \label{l12}
The constructed polygon is weakly visible from the edge $UV$.
\end{lemma}

\begin{proof}
It is easy to see that all the interior points of the polygon lying above the line $y=a+a'$, those lying between the lines $y=b+b'$ \& $y=a$, 
and also those lying between the lines $y=0$ \& $y=b$ are visible from the edge $UV$. The slope of the line $UV$, the choice of $U$ on it, 
and the way we set the value of $a'$ together ensure that, for every pair $(i,j)$ such that the spike corresponding to $e_j$ is visible from 
$s_i$, both the rays $\overrightarrow{D_j s_i}$ and $\overrightarrow{D'_j s_i}$ intersect $UV$. This implies that $UV$ sees all interior points 
within the cones formed by every such pair of rays, which includes every interior point of the polygon lying between successive holes in the 
barrier (i.e. between the lines $y=b$ \& $y=b+b'$), as well as every point lying within the spikes corresponding to the elements $e_j$ (i.e. 
lying below the line $y=0$). Finally, observe that for each $s_i$, the rays $\overrightarrow{s_i z_i}$ and $\overrightarrow{s_i z'_i}$, obtained 
by extending the two sides of the corresponding triangular hole, also intersect $UV$. Thus, it is guaranteed that $UV$ even sees all the interior 
points lying between successive triangular holes, i.e. between the lines $y=a$ \& $y=a+a'$, which was the only region not considered so far.
\end{proof}

\subsection{The reduction is polynomial}
Observe that $L$, $\theta$, $d$, $d'$, $d''$, $a$, $b$ are all constants in our reduction. The values for $a'$, $b'$, $x_u$, $y_u$ and 
every $D_j$ for $j=1,\dots,n$ are computable in polynomial time and can be expressed with $O(n \log m)$ bits. Moreover, the computation 
of all angles and intersection points required for the construction can be done in polynomial time. So, the construction of the weak 
visibility polygon produces a polynomial number of points each of which can be computed in polynomial time and take at most $O(n \log m)$ 
bits to be expressed. Therefore, it can be done in time polynomial in the size of the input Set Cover instance. Furthermore, it follows from 
Lemma \ref{l13} below that the transformation of an optimal solution for any Set Cover instance to an optimal solution for the corresponding 
Vertex Guard instance also takes polynomial time. 

\begin{figure}
\centerline{\includegraphics[height=171mm]{reduction_modified+details.pdf}}
\caption{The modified reduction for weak visibility polygons with holes 
         with inlay showing details for the construction of triangular holes corresponding to each $s_i$.}
\label{m+d}
\end{figure}

\begin{lemma} \label{l13}
In the construction in Section \ref{reduction_modified}, an optimal solution of size $k$ for a Set Cover instance induces an 
optimal solution of size at most $k+4$ for the corresponding Vertex Guard instance, whereas an optimal solution of size $k$ 
for a Vertex Guard instance induces an optimal solution of size at most $k-3$ for the corresponding Set Cover instance.
\end{lemma}

\begin{proof}
The choice of $U$, the slope of the line segment $UV$, and the choice of vertices $z_i$ and $z'_i$ for each triangular hole 
(corresponding to set $s_i$) together ensure the following -
\begin{itemize}
 \item Each interior point of the constructed polygon lying above the line $y=a+a'$ is visible from $U$.
 \item Each interior point of the polygon lying between the lines $y=a$ \& $y=a+a'$ is visible from $U$ or $Y$.
 \item Each interior point of the polygon lying between the lines $y=b+b'$ \& $y=a$ is visible from $Y$.
 \item Each interior point of the polygon lying between the lines $y=b$ \& $y=b+b'$ is visible from $U$, $P$ or $V$.
 \item Each interior point of the polygon lying between the lines $y=0$ \& $y=b$ is visible from both $P$ and $V$.
 \item Each interior point of the polygon lying below the line $y=0$ (i.e. the points belonging to the spikes corresponding to 
 each element $e_j$) is visible from at least one $s_i \in S'$ such that $S' \subseteq \{s_1,s_2,\dots,s_m\}$ is an optimal 
 solution of the Set Cover instance.
\end{itemize}
Therefore, given an optimal solution of size $k$ for any instance of Set Cover, we can construct an optimal set of size at most $k+4$ 
for the corresponding instance of Vertex Guard that consists of the vertices $P$, $V$, $U$, $Y$, along with every $s_i$ such that the 
set $s_i$ is part of the optimal solution for the Set Cover instance. On the other hand, any optimal solution of a Vertex Guard instance
must include the vertices $U$ and $Y$ (in order to guard interior points above the line $y=a+a'$, and between the lines $y=b+b'$ \& $y=a$, 
respectively), and at least one of $P$ and $V$ (in order to guard interior points between the lines $y=0$ \& $y=b+b'$), along with some 
subset $S' \subseteq \{s_1,s_2,\dots,s_m\}$. So, if the size of the optimal Vertex Guard solution is $k$, then $|S'| \leq k-3$, and $S'$ 
forms an optimal solution for the corresponding Set Cover instance.
\end{proof}

\subsection{An inapproximability result}
As mentioned in Section \ref{rhar}, Eidenbenz, Stamm and Widmayer \cite{ESW_2001} proved that, for polygons with holes, there cannot exist a 
polynomial time algorithm for the art gallery problem with an approximation ratio better than $((1−\epsilon)/12)\ln n$ for any $\epsilon>0$, 
unless NP $\subseteq$ TIME($n^{\mathcal{O}(\log\log n)}$). In order to prove this inapproximability result, they used a reduction from the 
Restricted Set Cover problem. We follow the same approach in order to establish our own inapproximability result for the case of polygons 
with holes that are weakly visible from an edge. \\

The Restricted Set Cover (RSC) problem consists of all Set Cover instances that have the property that the number of sets $m$ is less than 
or equal to the number of elements $n$, i.e. $m \leq n$. Eidenbenz, Stamm and Widmayer proved the following lemma.
\begin{lemma}[Lemma 9 in \cite{ESW_2001}] \label{RSC_quasi_hardness}
RSC cannot be approximated by any polynomial time algorithm with an approximation ratio of $(1−\epsilon)\ln n$ for every
$\epsilon>0$, unless NP $\subseteq$ TIME($n^{\mathcal{O}(\log\log n)}$).
\end{lemma}
A recent strengthening of Feige's \cite{Feige_1998} quasi-NP-hardness (see the notion of quasi-NP-hardness in \cite{AroraLund_1996}) 
result for Set Cover approximation by Dinur and Steurer \cite{DS_2014} allows us to invoke the stronger version below. 
\begin{lemma} \label{RSC_hardness}
RSC cannot be approximated by any polynomial time algorithm with an approximation ratio of $(1−\epsilon)\ln n$ for every
$\epsilon>0$, unless NP~=~P.
\end{lemma}

The modified reduction presented in Section \ref{reduction_modified} leads to the following lemma, similar to Lemma 10 in \cite{ESW_2001}.
\begin{lemma} \label{RSC_promise}
Consider the promise problem of RSC (for any $\epsilon > 0$), where it is promised that the optimum solution $OPT$ is either less than
or equal to $c$ or greater than $c(1-\epsilon)\ln n$ with $c$, $n$ and $OPT$ depending on the instance $I$. This problem is NP-hard.
Then, the optimum value $OPT'$ of the corresponding instance $I'$ of the Vertex Guard problem for polygons with holes that are weakly 
visible from an edge, is either less than or equal to $c+4$ or greater than $\frac{c+4}{12}\cdot(1-\epsilon)\ln|I'|$. More formally:
\begin{align}
 OPT \leq c                & \Rightarrow  OPT' \leq c+4                                  	   \label{eqn:1} \\ 
 OPT > c(1-\epsilon)\ln n  & \Rightarrow  OPT' > \frac{c+4}{12}\cdot(1-\epsilon)\ln|I'|          \label{eqn:2}
\end{align}
\end{lemma}               

\begin{proof}
The implication in \ref{eqn:1} follows trivially from Lemma \ref{l13}. 
We prove the contrapositive of \ref{eqn:2}, i.e.
$$ OPT' \leq \frac{c+4}{12}\cdot(1-\epsilon)\ln|I'| \Rightarrow OPT \leq c(1-\epsilon)\ln n $$ 
Recall from the proof of Lemma \ref{l13} that if we are given an optimal solution $OPT'$ of $I'$ with $k$ guards, it is guaranteed to contain the 
vertices $U$ and $Y$, and at least one of $P$ and $V$. So, we can obtain an optimal solution of $I$ with at most $k-3$ sets, simply by choosing 
$OPT = OPT' \setminus \{P,V,U,Y\}$. Therefore,
\begin{align}
 OPT  & \leq  \frac{c+4}{12}\cdot(1-\epsilon)\ln|I'| - 3  \\
      & \leq \frac{c+4}{12}\cdot(1-\epsilon)\ln n^3   \\
      & \leq \frac{4c}{12}\cdot3(1-\epsilon)\ln n   \\
      & \leq c(1-\epsilon)\ln n  
\end{align}
where we used $|I'| \leq n^3$ in (4), which is true because the polygon of $I'$ consists of $n$ spikes and less than $nm \leq n^2$ holes (since
$m<n$ in any instance of RSC), and therefore, the polygon consists of less than $k(n^2 +n)$ points, where $k$ is a constant.  
\end{proof}

\begin{theorem} \label{th_inapproximability}
For polygons with holes that are weakly visible from an edge, the Vertex Guard problem cannot be approximated by any polynomial time algorithm 
with an approximation ratio of $((1−\epsilon)/12)\ln n$ for every $\epsilon>0$, unless NP~=~P.  
\end{theorem}

\section{A 3-Approximation Algorithm for Placing Vertex Guards in Orthogonal Weak Visibility Polygons}
\label{algo3}

The class of orthogonal polygons weakly visible from an edge has been previously studied by Carlsson, Nilsson and Ntafos \cite{CNN_1993}
under the name of Manhattan skyline or histogram polygons, and they showed that there exists a linear time greedy algorithm to optimally
guard these polygons with point guards. Let us also consider a polygon $P$ belonging to this class, i.e. $P$ is an orthogonal polygon weakly 
visible from an edge $uv$. In this section, we present a simpler algorithm for vertex guarding $P$ with an approximation factor of 3 -- a 
clear improvement over the factor 6 which we obtained for the more general class of weak visibility polygons. \\

First, we present an algorithm for computing a guard set $S_A$ covering only the vertices of $P$, described below in pseudocode as Algorithm \ref{VG_orthogonal}.

\begin{algorithm}[H]
\caption{An $\mathcal{O}(n^2)$-algorithm for computing a guard set $S_A$ for all vertices of $P$}
\label{VG_orthogonal}
\begin{algorithmic}[1]
\State Compute $SPT(u)$ and $SPT(v)$ \label{VG_orthogonal:1} 
\State Initialize all the vertices of $P$ as unmarked \label{VG_orthogonal:2}
\State Initialize $A \leftarrow \emptyset$ and $S_A \leftarrow \emptyset$ \label{VG_orthogonal:3}
\While {there exist unmarked vertices in $P$} \label{VG_orthogonal:4}
\State $z \leftarrow u$  \label{VG_orthogonal:5}
\While{$z \neq v$} \label{VG_orthogonal:6}
\State $z \leftarrow$ the vertex next to $z$ in clockwise order on $bd_c(u,v)$ \label{VG_orthogonal:7} 
\If{$z$ is unmarked and $bd_c(p_u(z),p_v(z))$ are visible from $p_u(z)$ or $p_v(z)$} \label{VG_orthogonal:8}
\State $A \leftarrow A \cup \{z\}$ and $S_A \leftarrow S_A \cup \{p_u(z),p_v(z)\}$ \label{VG_orthogonal:9}
\State Place guards on $p_u(z)$ and $p_v(z)$ \label{VG_orthogonal:10}
\State Mark all vertices of $P$ that become visible from $p_u(z)$ or $p_v(z)$ \label{VG_orthogonal:11}
\EndIf \label{VG_orthogonal:12}
\EndWhile \label{VG_orthogonal:13}
\EndWhile \label{VG_orthogonal:14}
\State \Return the guard set $S_A$ \label{VG_orthogonal:15} 
\end{algorithmic}
\end{algorithm}

\begin{lemma} \label{l14}
Algorithm \ref{VG_orthogonal} always terminates. 
\end{lemma}

\begin{proof}
Termination is guaranteed by the dual properties of orthogonality and weak visibility.
\end{proof}

\begin{lemma} \label{l15}
Any guard $g \in S_{opt}$ that sees vertex $z$ of $P$ must lie on $bd_c(p_u(z),p_v(z))$. 
\end{lemma}

\begin{proof}
Since $p_u(z)$ is the parent of $z$ in $SPT(u)$, $z$ cannot be visible from any vertex of $bd_c(u,p_u(z))$. Similarly, 
since $p_v(z)$ is the parent of $z$ in $SPT(v)$, $z$ cannot be visible from any vertex of $bd_{cc}(v,p_v(z))$. Hence, 
any guard $g \in S_{opt}$ that sees $z$ must lie on $bd_c(p_u(z),p_v(z))$. 
\end{proof}

\begin{figure}[H]
\begin{minipage}{.55\textwidth}
\centerline{\includegraphics[height=58mm]{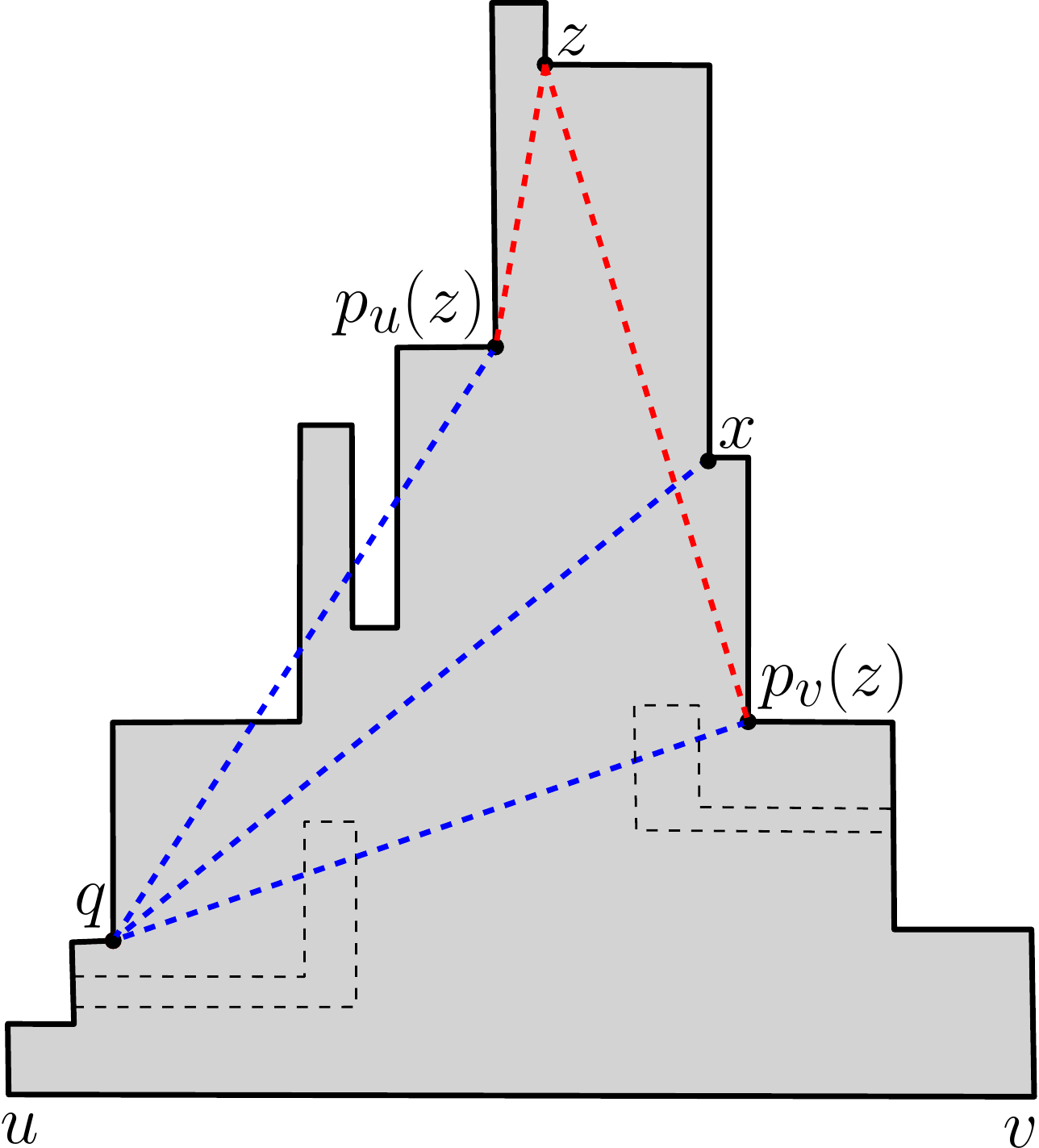}}
\caption{Case in Lemma \ref{l16} where $q$ lies on $bd_c(u,p_u(z))$.}
\label{orthogonal}
\end{minipage}
\quad
\begin{minipage}{.42\textwidth}
\centerline{\includegraphics[height=55mm]{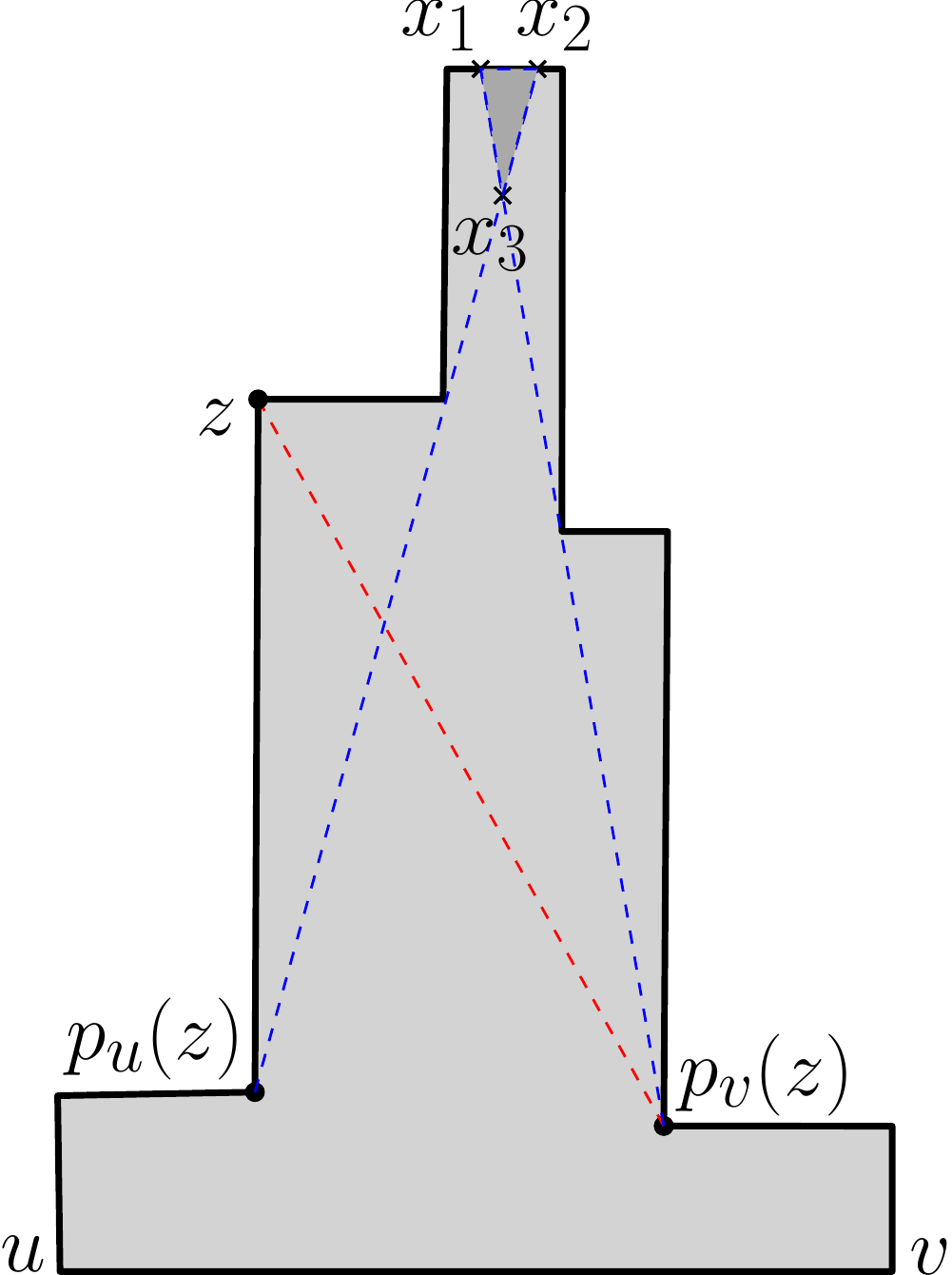}}
\caption{All vertices of the orthogonal polygon are visible from $p_u(z)$ or $p_v(z)$, but the triangle $x_1 x_2 x_3$ is invisible.}
\label{invisible}
\end{minipage}
\end{figure}

\begin{lemma} \label{l16}
Let $z \in A$. 
For every vertex $x$ lying on $bd_c(p_u(z),p_v(z))$, if $x$ sees a vertex $q$ of $P$, then $q$ must also be visible from $p_u(z)$ or $p_v(z)$. 
\end{lemma}

\begin{proof}
Since $z \in A$, $z$ must be a vertex of $P$ such that all vertices of $bd_c(p_u(z),p_v(z))$ are visible from $p_u(z)$ or $p_v(z)$. Hence,
if $q$ lies on $bd_c(p_u(z),p_v(z))$, then it is visible from from $p_u(z)$ or $p_v(z)$. So, consider the case where $q$ lies on 
$bd_{cc}(p_u(z),p_v(z))$. Now, either $q$ lies on $bd_c(u,p_u(z))$ or $q$ lies on $bd_{cc}(v,p_v(z))$. In the former case, 
if $bd_{cc}(q,p_v(z))$ intersects the segment $q p_v(z)$ (see Figure \ref{orthogonal}), then $q$ or $p_v(z)$ is not weakly visible from $uv$. Moreover, 
no other portion of the boundary can intersect $q p_v(z)$ since $qx$ and $z p_v(z)$ are internal segments. Hence, $q$ must be visible from $p_v(z)$. 
Analogously, if $q$ lies on $bd_{cc}(v,p_v(z))$, $q$ must be visible from $p_u(z)$.  
\end{proof}

\begin{lemma} \label{l17}
$|A| \leq |S_{opt}|$. 
\end{lemma}

\begin{proof}
Assume on the contrary that $|A| > |S_{opt}|$. This implies that Algorithm \ref{VG_orthogonal} includes two distinct vertices $z_1$ and $z_2$ 
belonging to $A$ which are both visible from a single guard $g \in S_{opt}$. Moreover, it follows from Lemma \ref{l15} that $g$ must lie 
on $bd_c(p_u(z_1),p_v(z_1))$. Without loss of generality, let us assume that vertex $z_1$ is added to $A$ before $z_2$ by Algorithm 
\ref{VG_orthogonal}. In that case, Algorithm \ref{VG_orthogonal} places guards at $p_u(z_1)$ and $p_v(z_1)$. Now, as vertex $z_2$ is visible 
from $g$, it follows from Lemma \ref{l16} that $z_2$ is also visible from $p_u(z_1)$ or $p_v(z_1)$. Therefore, $z_2$ is already marked, and 
hence, Algorithm \ref{VG_orthogonal} does not include $z_2$ in $A$, which is a contradiction.
\end{proof}

\begin{lemma} \label{l18}
$|S_A| = 2|A|$.
\end{lemma}

\begin{proof}
For every $z \in A$, since Algorithm \ref{VG_orthogonal} includes both the parents $p_u(z)$ and $p_v(z)$ of $z$ in $S_A$, it is clear that 
$|S_A| \leq 2|A|$. If both the parents of every $z \in A$ are distinct, then $|S_A| = 2|A|$. Otherwise, there exists two distinct vertices 
$z_1$ and $z_2$ in $A$ that share a common parent, say $p$. Without loss of generality, let us assume that vertex $z_1$ is added to $A$ 
before $z_2$ by Algorithm \ref{VG_orthogonal}. In that case, Algorithm \ref{VG_orthogonal} places a guard at $p$, which results in $z_2$ getting 
marked. Thus, Algorithm \ref{VG_orthogonal} cannot include $z_2$ in $A$, which is a contradiction. Hence, it must be the case that $|S_A| = 2|A|$.
\end{proof}

\begin{lemma} \label{l19}
$|S_A| \leq 2|S_{opt}|$. 
\end{lemma}

\begin{proof}
By Lemma \ref{l18}, $|S_A| = 2|A|$. By Lemma \ref{l17}, $|A| \leq |S_{opt}|$. So, $|S_A| = 2|A| \leq 2|S_{opt}|$.
\end{proof}



All interior points of $P$ are not guaranteed to be visible from guards in the set $S_A$ computed by Algorithm \ref{VG_orthogonal}. Consider the 
polygon shown in Figure \ref{invisible}. While scanning $bd_c(u,v)$, our algorithm places guards at $p_u(z)$ and $p_v(z)$ as all vertices of 
$bd_c(p_u(z),p_v(z))$ become visible from $p_u(z)$ or $p_v(z)$. Observe that in fact all vertices of $P$ become visible from these two guards. 
However, the triangular region $P \setminus (VP(p_u(z)) \cup VP(p_v(z)))$, bounded by the segments $x_1 x_2$, $x_2 x_3$ and $x_3 x_1$, is 
not visible from $p_u(z)$ or $p_v(z)$. Also, one of the sides $x_1 x_2$ of the triangle $x_1 x_2 x_3$ is a part of a polygonal edge. In fact, for any 
such region invisible from guards in $S_A$, one of the sides must always be a part of a polygonal edge. As mentioned previously in Section \ref{algo2},
any such region invisible from guards in $S$ is referred to as an \emph{invisible cell}, and the polygonal edge which contributes as a side to the 
invisible cell is referred to as its corresponding \emph{partially invisible edge}. Also, we define lid points and lid vertices as before. Next, we present 
an algorithm for computing an additional set of guards $S'_A$ whose placement ensures that all interior points of $P$ are also guarded.

\begin{algorithm}[H]
\caption{An $\mathcal{O}(n^2)$-algorithm for computing a guard set $S_A \cup S'_A$ for guarding $P$ entirely} 
\label{VG_orthogonal_final}
\begin{algorithmic}[1]
\State Compute $SPT(u)$ and $SPT(v)$ 
\State Compute the set of guards $S_A$ using Algorithm \ref{VG_orthogonal}.
\State Initialize $C \leftarrow \emptyset$, $S'_A \leftarrow \emptyset$ and $z \leftarrow u$ 

\While { there exists an edge in $P$ that is partially visible from guards in $S_A \cup S'_A$  }
\State $z' \leftarrow$ the vertex next to $z$ in clockwise order on on $bd_c(u,v)$

\If{ if the edge $zz'$ is partially visible from guards in $S \cup S'_A$ }
\State $c_i \leftarrow$ the lid point of the left pocket on $zz'$
\State $C \leftarrow C \cup \{c_i\}$ and $S'_A \leftarrow S'_A \cup \{p_u(c_i)\}$
\EndIf

\State $z \leftarrow z'$
\EndWhile

\State \Return the guard set $S_A \cup S'_A$
\end{algorithmic}
\end{algorithm}

\begin{theorem} \label{t20}
The running time of Algorithm \ref{VG_orthogonal_final} is $\mathcal{O}(n^2)$.   
\end{theorem}

\begin{proof}
$SPT(u)$ and $SPT(v)$ can be computed in $\mathcal{O}(n)$ time \cite{GHLST_1987}. Then, the computation of the guard set $S_A$ takes 
$\mathcal{O}(n^2)$ time, since it involves scanning the boundary of $P$ and identifying vertices to be marked whenever new guards are placed. 
The number of lid points on an edge can be at most $\mathcal{O}(n)$. Therefore, each time a new vertex is added to $S'_A$, the invisible portion of 
the first partially visible edge in clockwise order can be determined in $\mathcal{O}(n)$ time. Hence, the overall running time of Algorithm 
\ref{VG_orthogonal_final} is $\mathcal{O}(n^2)$. 
\end{proof}

We have the following lemma connecting $S'_A$ with $S_{opt}$.

\begin{lemma} \label{l21}
$|C| = |S'_A| \leq |S_{opt}|$.   
\end{lemma}



\begin{theorem} \label{t22}
$|S_A \cup S'_A| \leq 3|S_{opt}|$.   
\end{theorem}

\begin{proof}
By Lemma \ref{l19} and Lemma \ref{l21}, $|S_A \cup S'_A| \leq |S_A| + |S'_A| \leq 2|S_{opt}| + |S_{opt}| \leq 3|S_{opt}|$.  
\end{proof}

Therefore, Algorithm \ref{VG_orthogonal_final} is a 3-approximation algorithm for solving the problem of guarding orthogonal 
polygons that are weakly visible from an edge with minimum number of vertex guards.


\section{NP-Hardness for Point Guarding Polygons Weakly Visible from an Edge}
\label{pgwvp}
We prove that the Point Guard problem in polygons weakly visible from an edge is NP-hard by showing a reduction from the decision version of 
the minimum line cover problem (MLCP), which is defined as follows. Let $\mathcal{L} = \{l_1,\dots,l_n\}$ be a set of $n$ lines in the plane. 
Find a set $P$ of points, such that for each line $l \in \mathcal{L}$ there is a point in $P$ that lies on $l$, and $P$ is as small as possible. 
Let DLCP denote the corresponding decision problem, that is, given $\mathcal{L}$ and an integer $k>0$, decide whether there exists a line 
cover of size $k$. DLCP is known to be NP-hard \cite{MT_1982}. Moreover, MLCP was shown to be APX-hard \cite{BHN_2001,KAR_2000}.

\begin{figure}[H]
\centerline{\includegraphics[height=88mm]{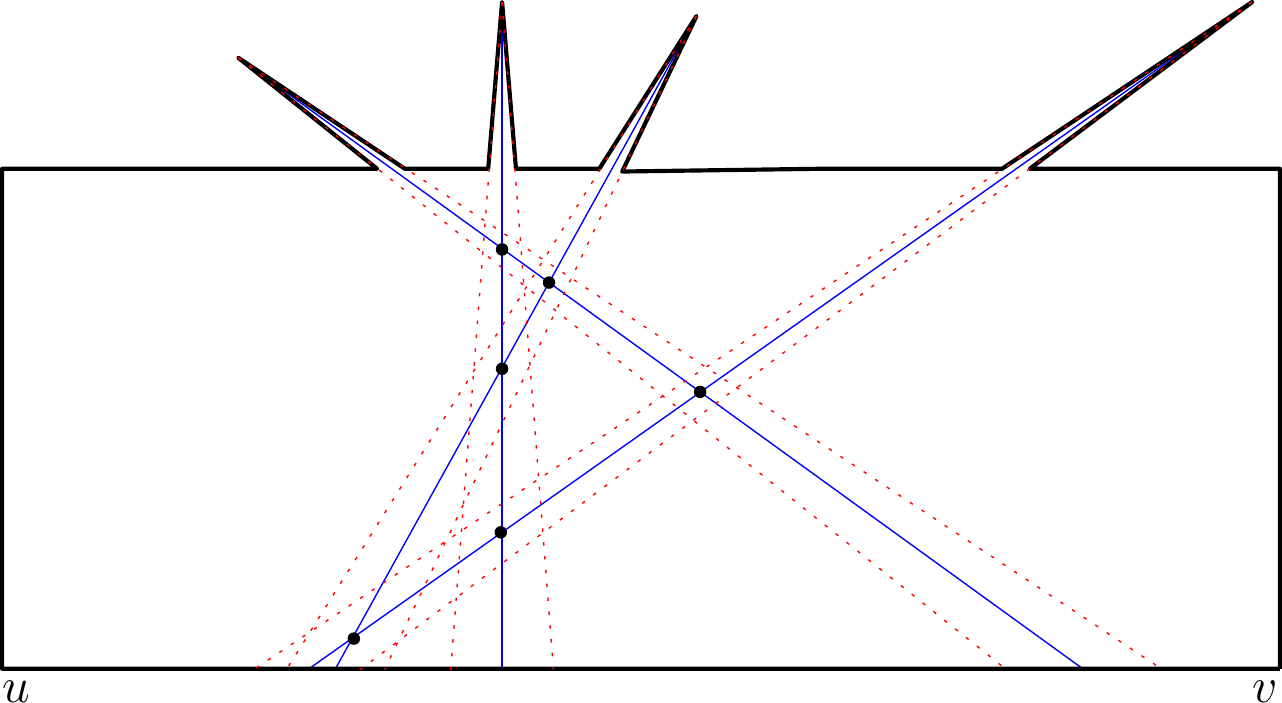}}
\caption{NP-hardness reduction from DLCP for point guarding polygons weakly visible from an edge}
\label{PG_DLCP}
\end{figure}

The reduction (see Figure \ref{PG_DLCP}) has the following steps. First, an axis-parallel rectangle $R$ is drawn on the plane such that it contains
all points of pairwise intersection of lines in $\mathcal{L}$. For each line $l \in \mathcal{L}$, consider the closed segment $l'$ that lies within 
this rectangle. Then, for each such segment $l'$, the end-point with the higher $y$ co-ordinate is extended beyond the boundaries of $R$ and 
a very narrow spike is added to the boundary of $R$ at this point. Note that, under this construction, the lower horizontal edge $uv$ of $R$ does 
not have any spikes added to it. In fact, the bounding rectangle along with the added spikes gives a polygon $P$ which is weakly visible from the 
edge $uv$. Let the tip of each spike be henceforth referred to as a \emph{distinguished point}. By making the spikes narrow enough, if it is ensured 
that the visibility polygons of no three distinguished points intersect, then the weak visibility polygon $P$ can be guarded using $k$ point guards if 
and only if the set of lines $\mathcal{L}$ has a cover of size $k$. One obvious way to achieve this correspondence is to restrict the placement of 
potential point guards to only the points of pairwise intersection of lines in $\mathcal{L}$. However, observe that instead of being placed exactly 
at the point of intersection of two lines $l_i,l_j \in \mathcal{L}$, a point guard can be placed (without losing any visibility) at any point within the 
intersection region of the visibility polygons of the distinguished points corresponding to the spikes generated by extending $l'_i$ and $l'_j$.       

\begin{theorem} \label{t23}
The Point Guard problem is NP-hard for polygons weakly visible from an edge.
\end{theorem}

\section{Acknowledgements}
The major part of this work was done when all the authors were in the School of Technology and Computer Science, Tata Institute of Fundamental 
Research, Mumbai-400005. Pritam Bhattacharya acknowledges the support provided to him by the TCS Research Scholar Program since July 2015. The 
authors would like to thank Girish Varma for pointing them towards a recent result on Set Cover approximation by Dinur and Steurer \cite{DS_2014}. 
The authors are also grateful to Sudebkumar P. Pal, Anil Maheshwari, and Bengt J. Nilsson for their many helpful comments which helped improve the
presentation of the paper.


\bibliographystyle{plain}
\bibliography{VGP_ref}

\end{document}